\documentclass[letterpaper]{scrartcl}
  
\usepackage[totalwidth=430pt, totalheight=600pt]{geometry}

\makeatletter
\renewcommand{\fnum@figure}{\textbf{Fig.~\thefigure}}
\makeatother

\makeatletter
\def\@fnsymbol#1{\ensuremath{\ifcase#1\or * \or \ddagger\or
   ~ \or \mathparagraph\or \|\or **\or \dagger\dagger
   \or \ddagger\ddagger \else\@ctrerr\fi}}
    \makeatother

\usepackage[english]{babel}
    
\usepackage{amsfonts}
\usepackage{amsmath}
\usepackage{amsthm}
\usepackage{amssymb}

\usepackage{microtype}
\usepackage{diagbox}

\usepackage{paralist}
\usepackage{algorithm,bbm}
\usepackage[noend]{algpseudocode}

\usepackage[font=small]{caption}
\usepackage{comment}
\usepackage{enumerate}
\usepackage{xfrac}
\usepackage{noindentafter}
\usepackage{stmaryrd}

\usepackage{thm-restate}

\usepackage[colorlinks=true,pagebackref=true]{hyperref}
\hypersetup{
    linkcolor = {blue},
    citecolor = {blue}
}
\usepackage{cleveref}

\usepackage{tikz}
\usetikzlibrary{shapes.misc}

\newcommand{\fO}{\mathcal{O}} %

\newcommand{\cA}{\mathcal{A}}
\newcommand{\rt}{\mathtt{root}} 
\newcommand{\depth}{\mathtt{depth}} 
\newcommand{\ch}{\mathtt{ch}} 
\newcommand{\cut}{\mathtt{cut}} 
\newcommand{\hgt}{\mathtt{depth}} 
\newcommand{\cost}{\mathtt{cost}} 
\newcommand{\OPT}{\mathtt{OPT}} %
\newcommand{\scl}{\mathtt{sc}} %
\newcommand{\pp}{\mathtt{pp}} %

\newcommand\TombStone{\rule{.7ex}{1.7ex}}
\renewcommand{\qedsymbol}{\TombStone}
\newcommand{\qedd}{\let\qed\relax\quad\raisebox{-.1ex}{$\qedsymbol$}}

\newcommand{\ceil}[1]{\lceil #1 \rceil}
\newcommand{\floor}[1]{\lfloor #1 \rfloor}
\newtheorem{lemma}{Lemma} %

\newtheorem{corollary}{Corollary}
\newtheorem{observation}{Observation}
\newtheorem*{openQuestion}{Open question}
\theoremstyle{definition}
\newtheorem{definition}{Definition}

\NoIndentAfterEnv{lemma}
\NoIndentAfterEnv{proposition}
\NoIndentAfterEnv{theorem}
\NoIndentAfterEnv{conjecture}
\NoIndentAfterEnv{corollary}
\NoIndentAfterEnv{observation}
\NoIndentAfterEnv{definition}
\NoIndentAfterEnv{openQuestion}

\NoIndentAfterEnv{restatable}
\NoIndentAfterEnv{proof}

\NoIndentAfterEnv{align*}
\NoIndentAfterEnv{enumerate}
\NoIndentAfterEnv{itemize}

\title{{Splay trees on trees}}

\def\titlefootnotetext{Institut f\"ur Informatik, Freie Universit\"at Berlin, \texttt{laszlo.kozma@fu-berlin.de}.\newline {Research supported by DFG grant KO 6140/1-1.}}

\author{
Benjamin Aram Berendsohn\thanks{Institut f\"ur Informatik, Freie Universit\"at Berlin, \texttt{beab@zedat.fu-berlin.de}}  \and 
L\'aszl\'o Kozma\thanks{\titlefootnotetext}}%

\begin{document}
\date{}
	\maketitle
	
	\begin{abstract}
		Search trees on trees (STTs) are a far-reaching generalization of binary search trees (BSTs), allowing the efficient exploration of tree-structured domains. (BSTs are the special case in which the underlying domain is a path.) Trees on trees have been extensively studied under various guises in computer science and discrete mathematics.
		
		Recently Bose, Cardinal, Iacono, Koumoutsos, and Langerman (SODA 2020) considered \emph{adaptive} STTs and observed that,
apart from notable exceptions, the machinery developed for BSTs in the past decades does not readily transfer to STTs. In particular, they asked %
		whether the optimal STT can be efficiently computed or approximated (by analogy to Knuth's algorithm for optimal BSTs), and whether natural self-adjusting BSTs such as Splay trees (Sleator, Tarjan, 1983) can be extended to this more general setting. 
		
		We answer both questions affirmatively. First, we show that a $(1 + \frac{1}{t})$-approximation of an optimal size-$n$ STT for a given search distribution can be computed in time $\fO(n^{2t + 1})$ for all integers $t \geq 1$. Second, we identify a broad family of STTs with linear rotation-distance, allowing the generalization of Splay trees to the STT setting. We show that our generalized Splay satisfies a \emph{static optimality} theorem, asymptotically matching the cost of the optimal STT in an online fashion, i.e.\ without knowledge of the search distribution. Our results suggest an extension of the \emph{dynamic optimality conjecture} for Splay trees to the broader setting of trees on trees.
		 
	\end{abstract}
	
	\newpage
	\tableofcontents
	\newpage
	
	\section{Introduction}\label{sec:intro}

Binary search trees (BSTs) support the efficient storage and retrieval of items from a totally ordered set, and are among the best-studied structures in computer science. The set of possible items, i.e.\ the ``search space'' of the BST is typically assumed to be a collection of integers. One may also take the collection of nodes on a path as the search space, with the obvious ordering along the path. %

This view suggests a broad generalization of BSTs, letting the underlying search space be, instead of a path, a \emph{general tree}. Denoting the underlying tree by $S$, the goal of a search is to locate a certain node $x$ of $S$. The search proceeds via oracle calls, where $x$ is ``compared'' to some node $y$ of $S$. The oracle either answers $x=y$ (in which case the search can stop), or identifies the connected component that contains node $x$, after the removal of node $y$ from $S$. The search then continues recursively within the identified connected component.

\begin{figure*}[h]
  \centering
  \includegraphics[width=8cm]{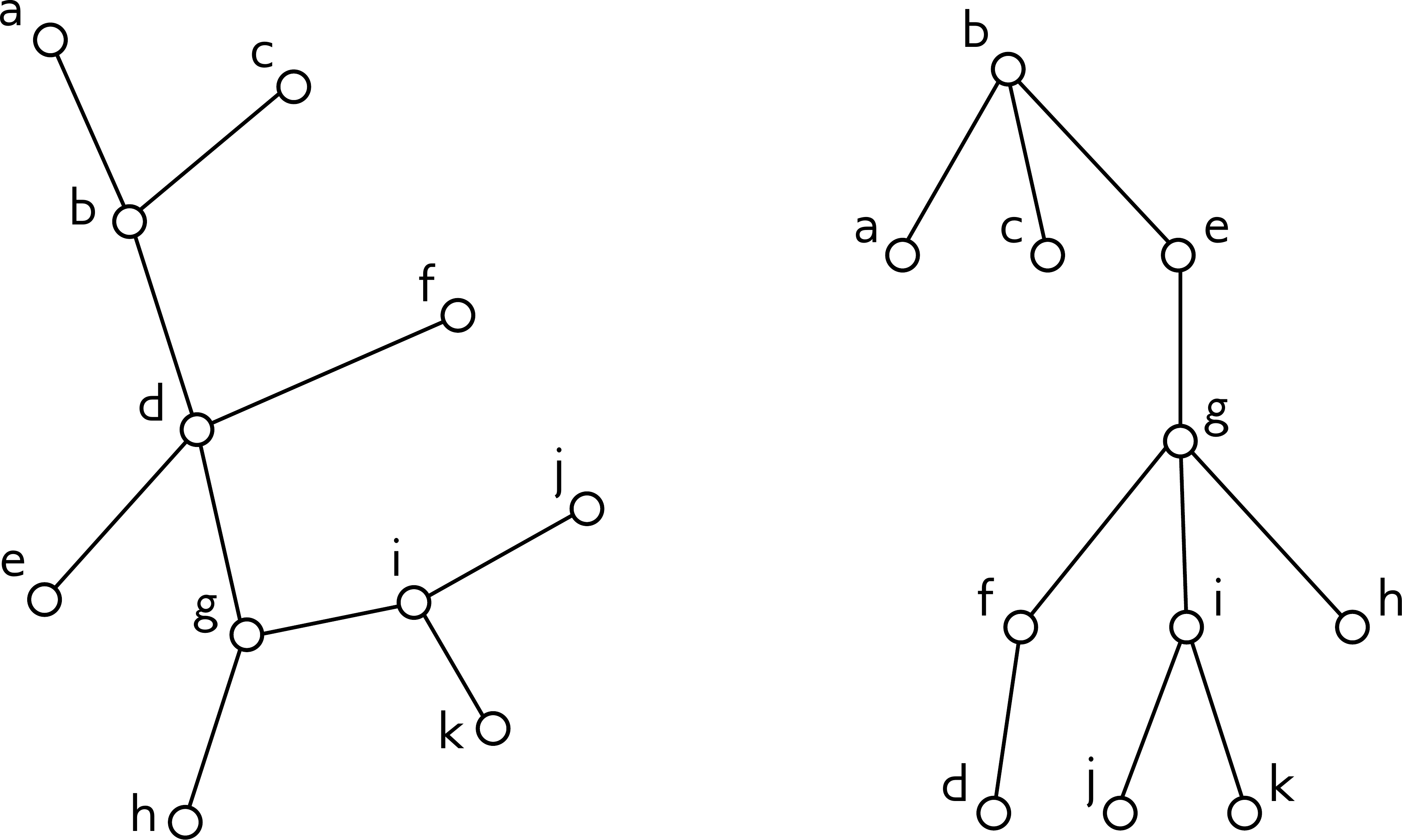}
  \caption{~~(\emph{left}) Tree $S$; ~~(\emph{right}) A search tree (STT) $T$ on $S$.\label{fig1}}
\end{figure*}

We may view a search strategy of this kind as a secondary tree $T$ on the nodes of $S$, built as follows. The root of $T$ is an arbitrary node $r$ of $S$, and the children of $r$ in $T$ are the roots of trees built recursively on the connected components of $S \setminus r$. We refer to such a tree $T$ as an STT (search tree on tree) on $S$. Oracle-calls are assumed constant-time, the search for $x$ is thus considered to take time proportional to the length of the search path from the root of $T$ to $x$. See \Cref{fig1} for an example. Note that $T$ is rooted, while $S$ is unrooted. Further note that the edge-sets of $T$ and $S$ may differ, and that the number of children of every node $x$ in $T$ is \emph{at most} the degree of $x$ in $S$. It follows that in the special case where $S$ is a path, the STT is, in fact, a BST. %

\medskip

One can further generalize STTs to allow searching in arbitrary graphs. Search trees on graphs (and trees) have been extensively studied in various settings.
Given a graph $G$, the minimal height of a search tree on $G$ is known as the \emph{treedepth} of $G$ (see e.g.\ \cite[\S\,6]{sparsity} for a comprehensive treatment). In other contexts, search trees on graphs have been studied as \emph{tubings}~\cite{carr2006coxeter}, \emph{vertex rankings}~\cite{deogun, bodlaender98, EvenSmorodinsky}, \emph{ordered colorings}~\cite{katchalski1995ordered}, or \emph{elimination trees}~\cite{liu1990role, pothen1990, AspvallHeggernes, bodlaender95} with applications in \emph{matrix factorization}, see e.g.~\cite[\S\,12]{duff2017}. In polyhedral combinatorics, search trees on trees are seen as vertices of a \emph{tree associahedron}, a special case of \emph{graph associahedra}~\cite{carr2006coxeter, devadoss2009realization, Postnikov}, and a generalization of the classical associahedron. %
The associahedron (whose vertices correspond to BSTs or other equivalent Catalan-structures) is a central and well-studied object of combinatorics and discrete geometry, see e.g.\ the recent book~\cite{muller2012associahedra} or survey~\cite{CSZ15} for a broad overview of this remarkable structure, its history, and further references.

Finding a search tree of minimal height on a graph is, in general, NP-hard~\cite{pothen1988}, but solvable in polynomial time for some special classes of graphs~\cite{AspvallHeggernes, deogun}. In particular, the minimum height search tree on a tree can be found in linear time by Sch\"affer's algorithm~\cite{schaeffer1989}, rediscovered multiple times during the past decades. An STT of logarithmic depth (analogously to a balanced BST) can be obtained via \emph{centroid decomposition}, an idea that goes back to the 19-th century work of Camille Jordan~\cite{Jordan1869}.

\medskip

In the context of searching, minimum height is a very limited form of optimality, only bounding the worst-case cost of a single search. For a given distribution of searches, the shape of the optimal tree may be very different from a minimum height tree. In the special case of BSTs, finding the optimal search tree for a given distribution is a well-understood problem. Knuth's textbook dynamic programming algorithm solves this task in $\fO(n^2)$ time~\cite{knuth_optimum}, and linear time constant-approximations have long been known~\cite{mehlhorn1975nearly, mehlhorn77abest}.

Recently, Bose, Cardinal, Iacono, Koumoutsos, and Langerman~\cite{Bose20} studied STTs in the context of searching, and explored whether the techniques developed for BSTs extend to STTs. They remark that no analogue of Knuth's algorithm is known for STTs, and it is not even clear whether the optimum search tree problem is polynomial-time solvable in this broader setting. Intuitively, the main difficulty is that, whereas BSTs consist of subtrees built over polynomially many candidate sets (corresponding to contiguous intervals of the search space), STTs consist of subtrees built \emph{over subtrees} of the search space, whose number is in general exponential. 

\medskip

Our first result (\S\,\ref{sec3}) is a polynomial-time approximation-scheme (PTAS) for the optimal STT problem. In the special case of $2$-approximating the optimum we obtain a simple $O(n^3)$ algorithm. To our knowledge, no constant-factor approximation was previously known. 

\begin{restatable}{theorem}{restatethma}\label{thm1}
Let $X$ be a search sequence over the nodes of an $n$-node tree $S$ and let $\OPT(X)$ be the minimum cost of serving $X$ in a fixed search tree on $S$. For every integer $t \geq 1$ we can find in time $\fO(n^{2t + 1})$ a search tree on $S$ that serves $X$ with cost at most $(1+\frac{1}{t}) \cdot \OPT(X)$.
\end{restatable}

The result combines a number of observations. The first, due to Bose et al., is that a restricted class of STTs they call \emph{Steiner-closed trees} (defined in \S\,\ref{sec2}) contains a tree of cost at most twice the optimum. In \S\,\ref{sec3} we generalize Steiner-closed trees to \emph{$k$-cut trees}, a family that approximates the optimum with arbitrary accuracy. We then show that, if we restrict attention to $k$-cut search trees for fixed $k$, then the number of admissible subproblems becomes polynomial, and an optimization similar to Knuth's algorithm can be carried out. 

\medskip

Optimal trees are still far from the full story of efficient search. Search trees can be re-structured between searches via \emph{rotations}, and in this way they can take advantage of various kinds of regularities in a search sequence. %
The most prominent \emph{adaptive} BST is the Splay tree, introduced by Sleator and Tarjan~\cite{ST85}. Splay trees perform local re-arrangements on the search path, with no apparent concern for global structure; data structures of this kind are also called \emph{self-adjusting}. Splay trees have powerful adaptive properties~\cite{ST85, tarjan_sequential, finger1, finger2, Sundar, deque_Pet08, LevyT19a, LevyT19b}, for instance, they asymptotically match the cost of the optimal tree, without a priori knowledge of the search distribution (a property known as \emph{static optimality}, shown by Sleator and Tarjan~\cite{ST85}). The stronger \emph{dynamic optimality conjecture} (one of the long-standing open questions of computer science) speculates that Splay trees are competitive with any self-adjusting strategy on any search sequence~\cite{ST85}. 

The dynamic optimality conjecture has inspired four decades of research, leading to powerful adaptive algorithms, instance-specific upper and lower bounds and structural insights about the BST model~(see \cite{in_pursuit, Kozma16, LevyT19a} for recent surveys). After Bose et al.\ initiated the study of adaptive STTs, it is very natural to ask to what extent this body of work can be transferred to the broader setting of STTs. 

\medskip

The rotation primitive readily extends from BSTs to STTs (\Cref{fig2}), and this opens the way for adaptive STT strategies. %
Bose et al.\ show that (surprisingly) a lower bound from the BST model due to Wilber~\cite{Wilber} can be extended to STTs. Building on this result, they obtain an STT analogue of Tango trees~\cite{tango}. Like Tango trees for BSTs, the structure of Bose et al.\ is $\fO(\log\log{n})$-competitive with the optimal \emph{adaptive} STT strategy, with $n$ denoting the number of nodes in the tree.

Bose et al.\ note several difficulties in achieving an arguably more natural goal: adapting \emph{Splay trees} to the STT setting. Conjectured to be $\fO(1)$-competitive, Splay trees are in many ways preferable to Tango trees. They have several proven distribution-sensitive properties (including static optimality) and are simple and efficient, both in theory and in practice. 

For BSTs, another well-studied adaptive strategy is Greedy, introduced independently by Lucas~\cite{Luc88} and Munro~\cite{Mun00}. Greedy can be viewed as a powerful \emph{offline} algorithm, that (essentially) re-arranges the search path in order of \emph{future} search times. Strikingly, Demaine et al.~\cite{DHIKP09} have shown that Greedy can be turned into an \emph{online} algorithm with only a constant-factor slowdown. %
More recently, with a better understanding of its behaviour, Greedy emerged as another promising candidate for dynamic optimality (see e.g.~\cite{FOCS15, LI16, GoyalG19}).

\medskip

There appears to be a major difficulty in transferring techniques from BSTs to STTs, in particular, in generalizing Splay and Greedy. An essential feature of the BST model is that any tree of size $n$ can be transformed into any other tree of size $n$ with $\fO(n)$ rotations~\cite{STT88, Pournin}. This fact affords a great flexibility in designing and analyzing algorithms, as the cost of restructuring a subtree can be charged to the cost of \emph{its traversal}, and the actual details of the rotations can be abstracted away. By contrast, as shown recently by Cardinal, Langerman and P\'{e}rez-Lantero~\cite{Cardinal18}, the rotation-diameter of STTs is $\Theta(n \log{n})$. This fact makes it unclear how direct analogues of Splay and Greedy may work in the STT model. %

We overcome this barrier by showing (\S\,\ref{sec4}) that the rotation-diameter is, in fact, linear, as long as 
we restrict ourselves to the already mentioned class of Steiner-closed trees. In fact, we bound the rotation-diameter for our more general class of $k$-cut trees.

\begin{restatable}{theorem}{restatethmb}
Given two $k$-cut search trees $T$ and $T'$ on the same $n$-node tree $S$, we can transform $T$ into $T'$ through a sequence of at most $(2k-1)n - (k+1)k + 1$ rotations. Moreover, starting with a pointer at the root, we can transform $T$ into $T'$ through a sequence of $\fO(k^2n)$ pointer moves and rotations at the pointer. All intermediate trees are $k$-cut trees.\label{thm3}
\end{restatable}

For BSTs, the fact that the rotation-diameter is linear can be shown by rotating both trees to a canonical \emph{path} shape~\cite{CulikWood}. Our proof of Theorem~\ref{thm3} can be seen as mimicking this classical argument, although the details are more subtle. The crucial observation is that, given two arbitrary $k$-cut trees, we can transform both trees to a \emph{canonical tree} that is a rooted version of the underlying search space $S$, using $\fO(kn)$ rotations.

The family of $k$-cut STTs (and Steiner-closed STTs, which correspond to the case $k=2$) thus form a connected, small-diameter core of tree associahedra, preserving useful properties of BSTs. By restricting ourselves to Steiner-closed trees, we regain part of the toolkit from BSTs. In particular, linear rotation distance allows us to implement natural transformations of the search path. 

\medskip

As our main result (\S\,\ref{sec5}), we define the SplayTT algorithm, a generalization of Splay to the setting of STTs. If the underlying search space $S$ is a path, SplayTT becomes the classical Splay tree. SplayTT keeps the search tree at all times in a Steiner-closed shape. Maintaining this property while restructuring the tree poses a number of technical difficulties. The main challenge is that the search ``path'', when viewed in the underlying search space, may (counter-intuitively) contain branchings of degree higher than two; this situation cannot arise in classical BSTs. %
We deal with this issue, by first splaying the higher degree branching nodes of the search path, followed by splaying the searched node itself.

We expect SplayTT to have distribution-sensitive properties that extend those of Splay trees to the STT setting. As a first result in this direction, we prove (\S\,\ref{sec6}) that SplayTT satisfies the analogue of \emph{static optimality} for Splay trees.

\begin{restatable}{theorem}{restatethmc}
Let $S$ be a tree of size $n$ and let $X$ be a sequence of $m$ searches over the nodes of $S$. Let $\OPT(X)$ denote the minimum cost of serving $X$ in a static search tree on $S$. Then the cost of SplayTT for serving $X$ is $\fO(\OPT(X) + n^2)$.  \label{thm4}
\end{restatable}

Despite the similarity between Splay and SplayTT, extending static optimality from Splay to SplayTT is not trivial. One of the obstacles already noted in~\cite{Bose20} is that \emph{Shannon entropy}, a natural measure of BST-efficiency cannot accurately capture the cost in the STT setting. The classical analysis of Splay trees via the \emph{access lemma}~\cite{ST85} appears closely tied to this quantity. To avoid this pitfall, we sidestep the access lemma and prove the static optimality of SplayTT directly, through a combinatorial argument.

We remark that the additive $\fO(n^2)$ term of Theorem~\ref{thm4} is independent of the length of the search sequence and depends only on the tree size. In fact, under the mild assumption that every node is searched at least once, the additive term can be removed. Theorem~\ref{thm4} strengthens Theorem~\ref{thm1}, in the sense that SplayTT does not know $X$ in advance. Just like in the BST setting, the static optimality of SplayTT also implies \emph{logarithmic amortized cost} (SplayTT is competitive with every tree, in particular, with the centroid decomposition tree). For STTs, however, static optimality is a significantly stronger claim; for any search distribution, the amortized cost is upper bounded by \emph{the treedepth} of the underlying search space, which may even be constant if e.g.\ $S$ is a star. %

\medskip
The strongest form of optimality for a self-adjusting search tree (and indeed, for any algorithm in any model) is \emph{instance-optimality}. In the case of search trees this is usually understood as matching the cost of the optimal adaptive strategy on every search sequence, up to some constant factor. As mentioned, this (conjectured) property of Splay trees is called \emph{dynamic optimality}. Since the generalization of Splay to SplayTT appears quite natural, we propose the following open question that subsumes classical dynamic optimality.

\begin{openQuestion}
Is SplayTT dynamically optimal in the STT model?
\end{openQuestion}

A further natural question is whether Greedy BST can be similarly  extended to the STT setting. The linear-time transformation between Steiner-closed trees (Theorem~\ref{thm3}) suggests a possible generalization, but its analysis appears to require the development of further tools, which we leave to future work. 

\paragraph{Further related work.} As mentioned, concepts related to STTs and more broadly to search trees on graphs have been studied in various contexts by different communities. Apart from the work of Bose et al.~\cite{Bose20}, that is closest in spirit to ours, earlier work has largely focused on \emph{minimum height}, i.e.\ the problem of computing the treedepth, or considered different models where queries are for \emph{edges}, see e.g.~\cite{BenAsher, LaberNogueira, MozesOnak, OnakParys}, or relatedly, searching in posets~\cite{LinialSaks, LinialSaks2, Carmo, Heeringa}.

In the edge-query model, search distributions (in the sense of Theorem~\ref{thm1}) have also been considered, and constant-approximations are known, with the existence of a PTAS remaining open~\cite{LaberMolinaro, Cicalese3, Cicalese1}. We remark that results do not directly transfer between the vertex-query and edge-query models. 

Other related, but not directly comparable work includes searching with weighted queries~\cite{Dereniowski}, searching with an oracle that identifies a shortest-path edge towards the target~\cite{EZK}, or searching with errors (stochastic or adversarial)~\cite{BorgstromKosaraju, FeigeRaghavan, KarpKleinberg, Finocchi, Boczkowski}. Search trees on graphs and trees have also been motivated with practical applications including file system synchronisation~\cite{BenAsher, MozesOnak}, software testing~\cite{BenAsher, MozesOnak}, asymmetric communication protocols~\cite{LaberMolinaro}, VLSI layout~\cite{Leiserson80}, and assembly planning~\cite{Iyer1}. Even and Smorodinsky~\cite{EvenSmorodinsky} relate STTs with the competitive ratio of certain online hitting set problems.

\section{Preliminaries}\label{sec2}

We use standard terminology on trees and graphs. A \emph{subtree} of an undirected tree is a connected subgraph. The set of nodes of a tree $S$ is denoted $V(S)$. The subgraph of $S$ induced by node set $A$ is denoted $S[A]$. By $S \setminus x$ we denote the forest obtained by deleting node $x$ in tree $S$. We say that $z$ \emph{separates $x$ and $y$}, if $x$ and $y$ fall into different connected components of $S \setminus z$, or equivalently if $z$ is on the path between $x$ and $y$ in $S$.
The \emph{convex hull} of a set of nodes $A \subseteq V(S)$, denoted $\ch(A)$ is the subtree of $S$ induced by the union of all paths between nodes in $A$. 

For a rooted tree $T$ and a node $x \in V(T)$ we denote by $T_x$ the subtree of $T$ rooted at $x$. 
The \emph{search path} of $x \in V(T)$ in $T$ is the unique path from the root to $x$. The number of nodes on the search path of $x$ in $T$ is $\hgt_T(x)$. Denoting the root of $T$ as $\rt(T)$, we have $\hgt_T(\rt(T)) = 1$. %

\paragraph{Search trees on trees.} We mostly follow the terminology of Bose et al.~\cite{Bose20}. %

\begin{definition}[Search tree on tree (STT)]
Given an unrooted tree $S$, a search tree on $S$ is a rooted tree $T$ with $V(T) = V(S)$, where the subtrees $T_x$ of $T$ are search trees on the connected components of $S \setminus \rt{(T)}$, for all children $x$ of $\rt(T)$. 
\end{definition}

Note that in STTs the ordering of children is irrelevant. See Figure~\ref{fig1} for illustration. The following observation is a direct consequence of the definition. 

\begin{observation}\label{obs1}
If $T$ is a search tree on $S$, then $S[V(T_x)]$ is a subtree of $S$ for all $x \in V(S)$. Furthermore, $T_x$ is a search tree on $S[V(T_x)]$.
\end{observation}

\noindent The rotation operation in STTs (see e.g.~\cite{Cardinal18}) generalizes BST rotations, see Figure~\ref{fig2}.

\begin{definition}[Rotation in STTs]
Consider a node $x$ with parent $p$ in a search tree $T$ on $S$. A \emph{rotation of the edge $\{x,p\}$}, alternatively called \emph{a rotation at node $x$} results in a tree  $T'$ obtained from tree $T$ as follows:
\begin{compactenum}[(i)]
\item $x$ and $p$ swap places, 
\item if $x$ has a child $y$ whose subtree $T_y$ contains a node adjacent to $p$ in $S$, then $y$ becomes the child of $p$,
\item all other children of $x$ and $p$ preserve their parent.
\end{compactenum}
\end{definition}

\begin{figure*}
  \centering
  \includegraphics[width=13cm]{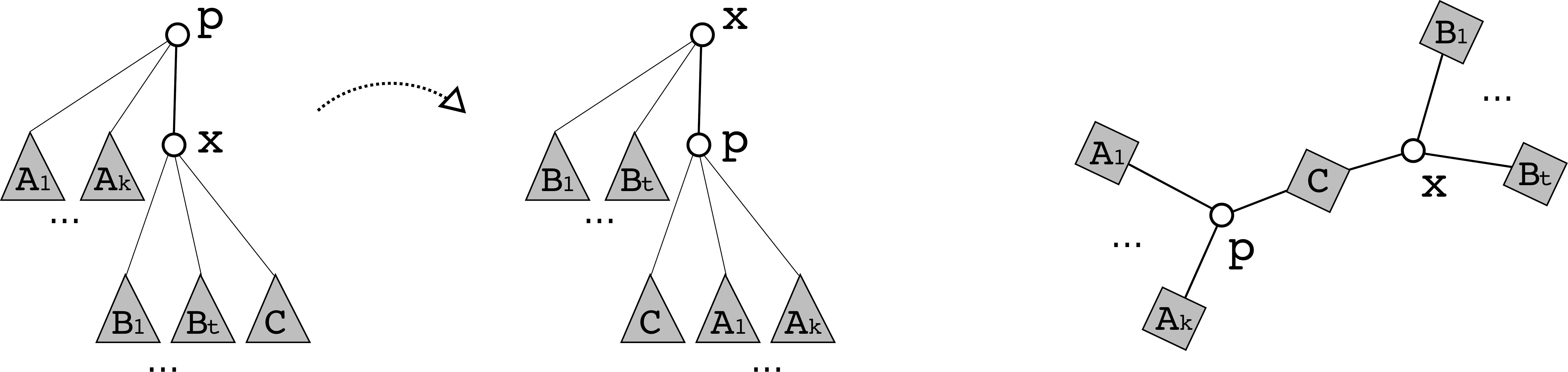}
  \caption{ Rotation of the edge $\{p,x\}$ in an STT $T$ (\emph{left}), and the underlying tree $S$ (\emph{right}). Triangles and diamonds represent subtrees of $T$ and $S$, respectively. Dots indicate an arbitrary number of subtrees that are affected in the same way by the rotation. One more rotation of the edge $\{p,x\}$ reverses the operation. \label{fig2}}
\end{figure*}

The following observation is immediate.

\begin{observation}
If $T$ is a search tree on $S$, then the tree $T'$ obtained from $T$ by an arbitary rotation is a search tree on $S$.
\end{observation}

\paragraph{Steiner-closed STT.} Bose et al.\ introduced an important property of STTs that also plays an essential role in our results. We review this concept next. 

\begin{definition}[Steiner-closed set~\cite{Bose20}]
A set of nodes $A \subseteq V(S)$ is Steiner-closed, if every node in $\ch(A) \setminus A$ is connected to exactly two nodes of $\ch(A)$.
\end{definition}

Observe that for all $A \subseteq V(S)$, the nodes in $\ch(A) \setminus A$ are connected to \emph{at least two} nodes of $\ch(A)$. Therefore, if $A$ is not Steiner-closed, then there is a node $q \in \ch(A) \setminus A$ that is connected to \emph{at least three} nodes of $\ch(A)$. See Figure~\ref{fig3} for illustration. The following observation is immediate. 

\begin{observation}\label{obspath}
If $\ch(A)$ is a path in $S$, then $A$ is a Steiner-closed set. 
\end{observation}

We next define Steiner-closed STTs. 

\begin{definition}[Steiner-closed STT~\cite{Bose20}]
An STT $T$ on $S$ is Steiner-closed, if for all $x \in V(S)$, the set of nodes in the search path of $x$ is a Steiner-closed set.
\end{definition}

\begin{figure*}
  \centering
  \includegraphics[width=3.5cm]{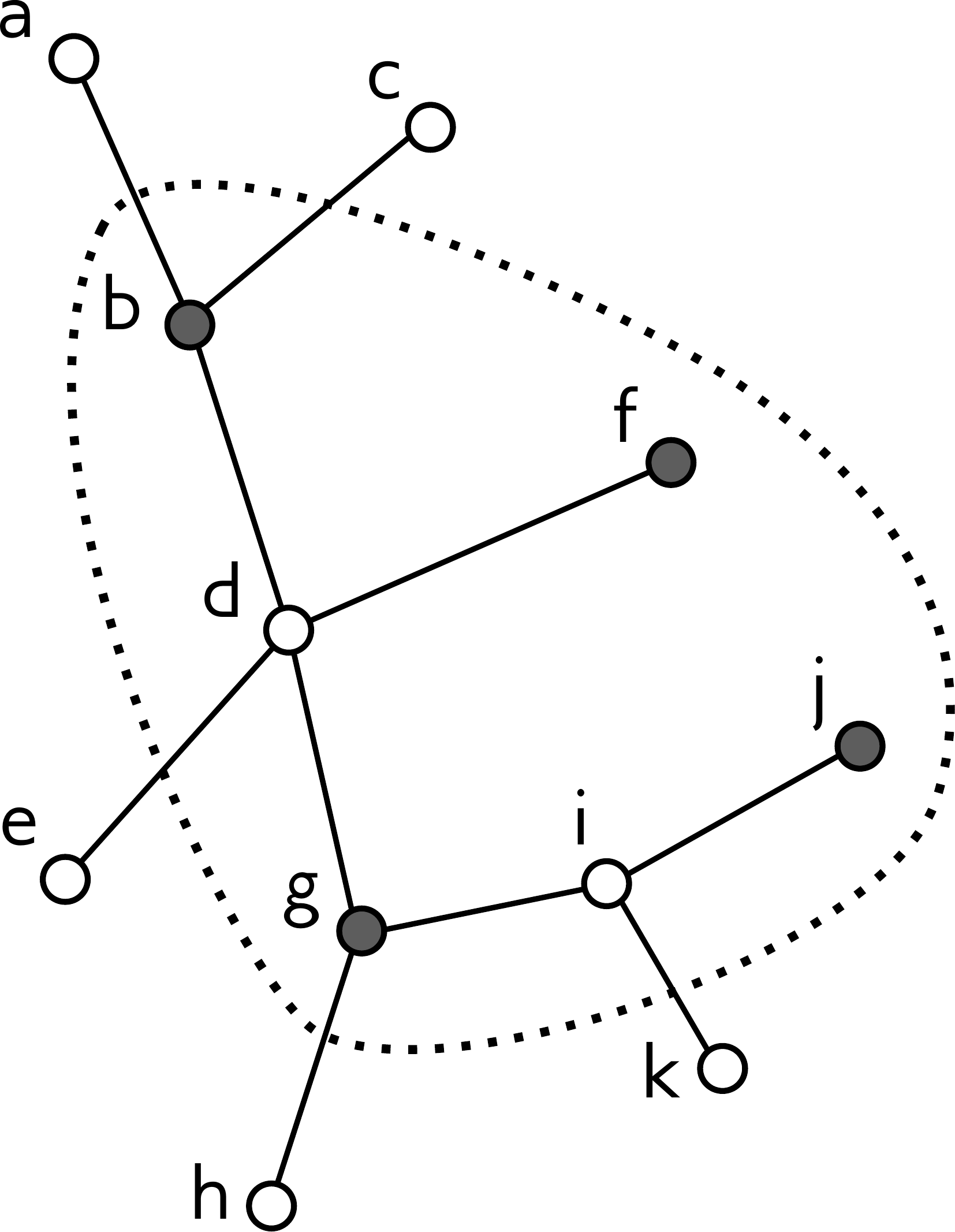}
  \caption{~~A tree $S$ with subset $A$ of nodes shaded and convex hull $\ch(A)$ shown with dotted line. Observe that $A$ is not Steiner-closed, but $A \cup \{d\}$ is Steiner-closed. \label{fig3}}
\end{figure*}

One can obtain a canonical Steiner-closed STT from the underlying search space itself.

\begin{observation} \label{obs6}
Let $S$ be an unrooted tree, and let $S^r$ be the rooted tree obtained by picking an arbitrary root $r \in V(S)$ in $S$. Then $S^r$ is a Steiner-closed STT on $S$.
\end{observation}

Steiner-closed STTs are useful in particular due to the following observation.

\begin{lemma}[{\cite[Lemma\ 4.2]{Bose20}}]
Given an STT $T$ on $S$, we can find, in polynomial time, a Steiner-closed STT $T'$ of $S$ so that $\depth_{T'}(x) \leq 2 \cdot \depth_{T}(x)$, for all $x \in V(S)$. \label{lemdepth}
\end{lemma}

Note that Lemma~\ref{lemdepth} is stated by Bose et al.\ for maximum depth, but it is explicitly observed in their proof that during the transformation from $T$ to $T'$ the depth of \emph{every node} at most doubles. An algorithm with running time $\fO(n^2)$ is implicit in the proof of Bose et al., with standard data structuring. We omit the proof of this statement, as we show a more general result in \S\,\ref{sec3}. More precisely, in \S\,\ref{sec3} we define the class of \emph{$k$-cut} STTs, and in \S\,\ref{sec:Steiner-2cut-equiv} we show that the $2$-cut and Steiner-closed properties are equivalent.

\paragraph{Static and dynamic STT model.} We now discuss the cost model of STTs, as a straightforward extension of the BST cost model (see e.g.\ \cite{Wilber, tango}). Let $T$ be an STT on $S$. The cost of searching $x \in V(S)$ in $T$ is $\depth_T(x)$. The cost of serving a sequence of searches $X = (x_1, \dots, x_m) \in V(S)^m$ in $T$ is $\cost_T(X) = \sum_{i=1}^m{\depth_T(x_i)}$. 

If re-arrangements of the tree $T$ are allowed, the model is as follows. An algorithm $\cA$ starts with an initial search tree $T_0$  on $S$, and $T_i$ denotes the state of the tree after the $i$-th search. At the start of the $i$-th search, a \emph{pointer} is at the root of $T_{i-1}$ and $\cA$ can perform an arbitrary number of steps of (1) rotating the edge between the node $x$ at the pointer and its parent, and (2) moving the pointer from the current node $x$ to its parent or to one of its children. When serving the search $x_i$, the pointer must visit, at least once, the node $x_i$. 

Both types of operations have the same unit cost. An additional unit cost is charged for performing each search. The cost of an algorithm $\cA$ for executing $X$, denoted $\cost_{\cA}(X)$, is thus the total number of pointer moves and rotations plus an additive term $m$. An algorithm is \emph{offline} if it knows the entire sequence $X$ in advance, and \emph{online} if it receives $x_i$ only after the $(i-1)$-th search has finished.

In both the static and the dynamic case we only account for operations \emph{in the model}. Algorithms that are to be considered efficient should, however, also spend polynomial time \emph{outside the model} (i.e.\ for deciding which rotations and pointer moves to perform).\footnote{Whether unbounded computation outside the model can improve the competitiveness of online algorithms is an intriguing theoretical question for both BSTs and STTs.} In case of our generalization of Splay, the time spent outside the model is \emph{linear} in the model cost.

\section{Almost optimal search trees on trees}\label{sec3}

Let $S$ be a tree on $n$ nodes and consider a search sequence $X=(x_1, \dots, x_m) \in V(S)^m$ with the function $p: V(S) \rightarrow \mathbb{N}$ denoting the frequencies of searches, so that each node $x \in V(S)$ appears $p(x)$ times in $X$. We want to find a search tree $T$ on $S$ in which $X$ is served with the smallest possible cost, i.e.\ to minimize $$\cost_T(X) ~=~ \sum_{i =1}^{m}{\hgt_T(x_i)} ~= \sum_{x \in V(S)}{p(x) \cdot\hgt_T(x)}.$$ %

In this section we show how to find, in polynomial time, a search tree $T'$ on $S$, whose cost is $\cost_{T'}(X) \leq (1+\varepsilon) \cdot \cost_T(X)$, for arbitrarily small $\varepsilon > 0$, i.e.\ we give a polynomial-time approximation scheme (PTAS) for the optimal STT problem. The result is based on $k$-cut trees, a generalization of Steiner-closed trees that we introduce (Steiner-closed trees are the special case $k=2$). Before presenting the algorithm, we need some definitions.

\paragraph{Cuts and boundaries.}

The \emph{cut} in $S$ of a nonempty set of nodes $A \subseteq V(S)$, denoted $\cut_S(A)$ or simply $\cut(A)$ is the set of (directed) pairs of nodes $(u,v)$, where $u \in V(S) \setminus A$, and $v \in A$, and $\{u,v\}$ is an edge of $S$. In words, the cut is the set of edges connecting the remainder of the tree $S$ to $A$, indicating the direction. Observe that $\cut(A) = \emptyset$ if and only if $A = V(S)$. Moreover, $\cut(A)$ uniquely determines the set $A$, and given $\cut(A)$, we can find $A$ through a linear time traversal of $S$.

The \emph{boundary} $\delta_S(A)$ or simply $\delta(A)$ %
is the set of nodes outside $A$ that define the cut. More precisely, $u \in \delta_S(A)$ if and only if $(u,v) \in \cut_S(A)$, for some $v \in A$.
Observe that if $S[A]$ is connected (i.e.\ a subtree), then $|\delta_S(A)| = |\cut_S(A)|$. We call this quantity the \emph{boundary size} of $A$.  To simplify notation, for subtrees $H$ of $S$ we let $\delta(H)$ denote $\delta(V(H))$.

\begin{definition}[$k$-cut tree]
	For $k \geq 1$, an STT $T$ on $S$ is a $k$-cut tree, if the boundary size $|\delta_S(T_x)|$ is at most $k$ for all nodes $x \in V(S)$. %
\end{definition}

It is easy to verify that $1$-cut STTs are exactly those obtained by rooting $S$ at some vertex (as in \Cref{obs6}). A more involved argument (\S\,\ref{sec:Steiner-2cut-equiv}) shows that $2$-cut trees are exactly the \emph{Steiner-closed trees}. Note that these equivalences directly imply \Cref{obs6}, since, by definition, every $2$-cut tree is also a $1$-cut tree.

As the number of possible cut edges in a tree is $\fO(n)$, the following observation is immediate, implying that the number of possible subtrees of a $k$-cut tree is polynomial, rather than exponential. 

\begin{observation}\label{p:num-k-cut-subsets}
	The number of subsets $A \subseteq V(S)$ with boundary size at most $k$ is $\fO(n^k)$.
\end{observation}

The following technical lemma relates the boundary sets before and after the removal of a node, and will be useful in the remainder of the section.

\begin{lemma}\label{p:comp-bounds}
	Let $S$ be a tree, let $S[A]$ be a subtree of $S$, and let $r \in A \subseteq S$. Let $N(r)$ be the set of neighbors of $r$ in $S$, and let $C_1, C_2, \dots, C_t$ be the connected components of $S[A] \setminus r$. Then:
	\begin{align*}
	\bigcup_{i=1}^t \delta(C_i) =  \delta(A) \cup \{r\} \setminus N(r).
	\end{align*}
\end{lemma} 
\begin{proof}
	Let $v \in \delta(C_i)$ for some $i \in [t]$. Then there is an edge $\{v, u\}$ with $u \in V(C_i)$. Either $v=r$, or the unique path from $v$ to $r$ in $S$ contains $u$, and therefore $v \notin N(r)$. If $v \neq r$, then %
	$v \notin C_j$ for all $j \neq i$, as otherwise $C_i$ and $C_j$ were connected in $S[A] \setminus r$. It follows that $(v,u) \in \cut(A)$, so $v \in \delta(A)$.  %

	Conversely, let $v \in  \delta(A) \cup \{r\} \setminus N(r)  $. If $v = r$, then $v \in \delta(C_i)$ for all $i$. Otherwise, let $v \in \delta(A) \setminus N(r)$, and let $v' \in A$ such that $(v,v') \in \cut(A)$. By assumption, $v' \neq r$. Let $C_i$ be the connected component of $S[A] \setminus r$ that contains $v'$. Then $v \in \delta(C_i)$. \qedd
\end{proof}

Observe that \Cref{p:comp-bounds} directly implies that if $x$ is a child of $p$ in $T$, then $\delta(T_x) \subseteq \delta(T_p) \cup \{p\}$. Consequently, the boundary $\delta(T_x)$ of a node $x$ consists entirely of ancestors of $x$ in $T$.

The rest of this section is dedicated to the proof of \Cref{thm1} and is organized as follows. In \S\,\ref{sec:k-cut-approx} we show that an optimal $k$-cut STT approximates an optimal general STT by a factor of roughly $1 + \frac{2}{k}$. In \S\,\ref{sec:k-cut-opt-find} we generalize the dynamic programming algorithm for BSTs and show that an optimal $k$-cut STT can be found in polynomial time. 

\subsection{\texorpdfstring{$k$-cut}{k-cut} trees approximate depth}\label{sec:k-cut-approx}

In this subsection we show that an arbitrary STT $T$ can be transformed into a \emph{a $k$-cut STT} $T'$, so that the depth of every node increases by a factor of no more than (roughly) $1+\frac{2}{k}$. The proof is based on a similar idea as the proof of \Cref{lemdepth} by Bose et al.~\cite{Bose20}: problematic nodes are fixed one-by-one, carefully controlling the depth-increase of every node. The cases $k>2$ require however some further ideas, in particular, we make use of the \emph{leaf centroid} of a tree, defined next.

\begin{definition}[Leaf centroid]
	Let $S$ be a tree on $n \ge 3$ nodes, having $\ell$ leaves. A non-leaf node $v \in V(S)$ is a \emph{leaf centroid} of $S$ if every connected component of $S \setminus v$ has
	at most $\floor{\frac{\ell}{2}}+1$ leaves, %
	at most $\floor{\frac{\ell}{2}}$ of which are leaves of $S$.
\end{definition}

The existence of a leaf centroid follows by a similar argument as the existence of the classical centroid (see e.g.~\cite{Slater1978,BodlaenderEtAl1992,Wang}). (Informally) start at an arbitrary non-leaf node $x$ and as long as $S \setminus x$ has a component with too many leaves, move $x$ along the edge towards this component. %
By standard data structuring, a leaf centroid can be found in linear time; for completeness we give a proof in Appendix~\ref{appa}. %

The following observation connects leaf-sets with boundaries.

\begin{observation}\label{p:delta-ch-leaves}
	Let $S[A]$ be a subtree of $S$. Then the set of leaves of $\ch(\delta_S(A))$ is $\delta_S(A)$.
\end{observation}
\begin{proof}
	Let $x \in \delta(A)$, and suppose that $x$ is not a leaf of $\ch(\delta(A))$. Then $x$ is on the path between two nodes $u, v \in \delta(A) \setminus \{x\}$. Let $(u,u'), (v,v') \in \cut(A)$. As $S[A]$ is connected, there is a path between $u'$ and $v'$ that lies completely within $S[A]$. The path between $u$ and $v$ consists of exactly this path, with $u$ prepended and $v$ appended. This means that $x \in A$, a contradiction. %
	
	Conversely, let $x \notin \delta(A)$. If $x \in \ch(\delta(A))$, then $x$ must lie on a path between two nodes $u, v \in \delta(A)$. This means that $x$ is not a leaf of $\ch(\delta(A))$. \qedd
\end{proof}

We proceed with the main lemma of this subsection.

\begin{lemma}\label{p:depth_approx}
	Given a search tree $T$ on $S$, for arbitrary $k \ge 3$ we can find in time $\fO(n^2)$ a $k$-cut search tree $T^*$ on $S$, so that $\depth_{T^*}(x) \le (1+\epsilon_k) \cdot \depth_T(x)$ for all $x \in V(S)$, where
	\begin{align*}
	\varepsilon_k = \frac{1}{\left\lceil \frac{k}{2} \right\rceil - 1}.
	\end{align*}
\end{lemma}

\Cref{alg:trans-st} transforms $T$ into $T^*$ (with the call $\Call{Fix}{T, r}$, where $r = \rt(T)$). The basic idea is the following: for a given node $x$ of $T$ (initially the root), we check whether $T_x$ has a boundary size smaller than $k$. %
If yes, we simply recurse on the subtrees. Otherwise, we transform $T_x$ by replacing the root $x$ with a node $v$ (by rotating $v$ to the top), such as to minimize the maximum boundary size of the subtrees rooted at the children. 
After the transformation, we recurse on the children of the new tree. (Note that when the boundary size is exactly $k$, node $v$ may happen to be $x$ in which case no rotation is necessary.)
 
 The rest of this subsection is dedicated to the proof of \Cref{p:depth_approx}. We remark that we only need Algorithm~\ref{alg:trans-st} as an \emph{existence proof} for good $k$-cut trees, and its running time does not affect the running time of our approximation algorithm. %

\begin{algorithm}[tbp]
	\caption{Transforming an arbitrary STT into a $k$-cut STT.}\label{alg:trans-st}
	\begin{algorithmic}[1]
		\Statex \textbf{Input:} search tree $T$ on $S$, constant $k \ge 3$, node $x \in V(T)$.
		\Procedure{Fix}{$T, x$}
			\If{$x$ is a leaf}
				\State \Return $T_x$
			\ElsIf{$|\delta(T_x)| < k$} %
				\State $c_1, c_2, \dots, c_t \gets$ children of $x$ in $T$\label{alg-line:just-recurse-1}
				\State \Return tree rooted at $x$ with subtrees $\Call{Fix}{T, c_1}, \dots, \Call{Fix}{T, c_t}$\label{alg-line:just-recurse-2}
			\ElsIf{$|\delta(T_x)| \ge k$} 
				\State $v \gets $ leaf centroid of $\ch(\delta(T_x))$\label{alg-line:leaf-centroid}
				\If{$v \neq x$}
					\State Obtain $T'$ from $T$ by rotating $v$ up to become the parent of $x$. \label{alg-line10}
				\Else
					\State $T' \gets T$
				\EndIf
				\label{alg-line:constr-T'}
				\State $c_1', c_2', \dots c_t' \gets$ children of $v$ in $T'$
				\State \Return tree rooted at $v$ with subtrees $\Call{Fix}{T', c_1'}, \dots, \Call{Fix}{T', c_t'}$
			\EndIf
		\EndProcedure
	\end{algorithmic}
\end{algorithm}

\paragraph{Correctness.}
We consider the boundary size of the subtrees in each recursive call.

\begin{lemma}\label{p:cutsize-inc}
	Let $x \in V(S)$ and let $c$ be a child of $x$ in $T$. Then $|\delta(T_c)| \le |\delta(T_x)| + 1$.
\end{lemma}
\begin{proof}
	By \Cref{p:comp-bounds}, we have $\delta(T_c) \subseteq \delta(T_x) \cup \{x\}$, so $|\delta(T_c)| \le |\delta(T_x)| + 1$.\qedd
\end{proof}

\begin{lemma}\label{p:cutsize-halved}
	Let $x \in V(S)$ with $|\delta(T_x)| \ge k$, let $T'$ be the tree produced in \Cref{alg-line10} or \Cref{alg-line:constr-T'} of \Cref{alg:trans-st}, let $v$ be the leaf centroid of $\ch(\delta(T_x))$ and let $c$ be a child of $v$ in $T'$. Then
	\begin{align*}
		|\delta(T'_{c})| \le \left\lfloor \frac{|\delta(T_x)|}{2} \right\rfloor + 1.
	\end{align*}
\end{lemma}
\begin{proof}
	Observe that $V(T_x) = V(T'_v)$.
	The set $V(T'_{c})$ is a connected component of the forest $S[V(T_x)] \setminus v$.
	Let $U$ be the set of nodes $u$ such that $(u,w) \in \cut(V(T_x))$ for some $w \in V(T'_{c})$. Each $u \in U$ is contained in $\delta(T_x)$, so it is a leaf of $S[\ch(\delta(T_x))]$ by \Cref{p:delta-ch-leaves}. Moreover, all $u \in U$ are in the same component of $S[\ch(\delta(T_x))\setminus v]$ (the one that contains $c$). As $v$ is a leaf centroid of $S[\ch(\delta(T_x))]$, we have $|U| \le \floor{|\delta(T_x)|/2}$.
	
	Finally, observe that $\delta(T'_c) = U \cup \{v\}$ by \Cref{p:comp-bounds} and the definition of $U$, so $|\delta(T'_{c})| = |U|+1$. \qedd
\end{proof}

From \Cref{p:cutsize-inc,p:cutsize-halved} and the fact that $|\delta(T)| = 0$, it follows inductively that in each recursive call $\Call{Fix}{T, c}$, the set $V(T_c)$ has boundary size at most $k$ (as $\floor{\frac{k}{2}}+1 \le k$).

\paragraph{Depth increase.}
We now bound the increase in depth due to the transformation for each node in $T$. %
Intuitively, when following the search path in the resulting tree $T^*$, we have a newly added node (compared to $T$) whenever the boundary size of the current tree is $k$, which by \Cref{p:cutsize-halved} can only happen every $\ceil{\frac{k}{2}}-1$ steps. We proceed with the formal proof.

Let $T$ be a search tree on $S$ with root $r$ and let $T^* = \Call{Fix}{T, r}$. Let $P$ be the search path in $T$ of an arbitrary node $u$ and let $P' = (u_1, u_2, \dots, u_t = u)$ be the search path of $u$ in $T^*$. Let $i_1 < i_2 < \dots < i_s$ be the indices of nodes in $P'$ that are not in $P$. %
We want to show that the number of such nodes is
\begin{align*}
s ~\le~ \varepsilon_k |P| ~=~ \frac{|P|}{\ceil{k/2}-1}.
\end{align*}
As $\depth_{T^*}(u) - \depth_{T}(u) = s$ and $|P| = \depth_{T}(u)$, this shows the bound stated in Lemma~\ref{p:depth_approx}.

\begin{lemma}
	$i_j + \ceil{\frac{k}{2}} \le i_{j+1}$ for all $j \in [s-1]$.
\end{lemma}
\begin{proof}
	Let $c_i = |\delta(T'_{u_i})|$ for $i \in [t]$. As $u_{i_j}$ is not in $P$, at some point, Algorithm~\ref{alg:trans-st} must have rotated $u_{i_j}$ up. This means that in some recursive call, $u_{i_j}$ is the leaf centroid $v$ that is rotated up in \Cref{alg-line10} and, in particular, $c_{i_j} = k$. Similarly, $c_{i_{j+1}} = k$. Let $\ell = i_{j+1} - i_j$.
	
	As $u_{i_j + 1}$ is a child of $u_{i_j}$ in $T^*$, we have $c_{i_j+1} \le \floor{k/2}+1$ by \Cref{p:cutsize-halved}. By \Cref{p:cutsize-inc}, we generally have $c_{i+1} \le c_i + 1$ for each $i \in [t-1]$. This means that $k = c_{i_{j+1}} = c_{i_j+\ell} \le \floor{k/2} + \ell$, which implies that $\ell \ge \ceil{k/2}$. \qedd
\end{proof}

As $c_1 = 0$, we also have $i_1 \ge k+1 \ge \ceil{k/2}$. As such, we can uniquely assign $\ceil{k/2}-1$ ``direct predecessor'' nodes in $P$ to each node in $P' \setminus P$. This proves the upper bound for $s$.

\paragraph{Running time.} In each recursive call we compute the boundary size of $V(T_x)$ in linear time, e.g.~by finding the ancestors of $x$ in $T$ and then traversing $V(T_x)$ from $x$. Furthermore, we may rotate one node to the root (of $T_x$), which requires linear time. As each node $v$ corresponds to exactly one recursive call (which returns a subtree rooted at $v$), the total running time is $\fO(n^2)$. This concludes the proof of \Cref{p:depth_approx}.

\paragraph{Remark.} When $k$ is even, \Cref{p:depth_approx} can be slightly strengthened to obtain $\varepsilon_k = \frac{1}{\lfloor {k}/{2} \rfloor}$, at the cost of a slightly more involved procedure. In particular, this extends the statement to the $k=2$ case. Without the improvement, the running time stated in Theorem~\ref{thm1} would be $\fO(n^{2t+2})$ instead of $\fO(n^{2t+1})$.

Intuitively, the improvement comes from the observation that the root-replacement of \Cref{alg-line10} in Algorithm~\ref{alg:trans-st} is too ``proactive''. When $T_x$ matches the boundary size condition with equality, it may be too early to rotate the replacement-root $v$ to the top, as the boundary size may recover as we go further down, if $x$ happens to split the tree in a reasonably balanced way. 
We defer the details of this small improvement to Lemma~\ref{p:depth_approx-improved}, Appendix~\ref{appb}.
With the improved bound and the observation that $2$-cut trees are exactly the Steiner-closed trees (\S\,\ref{sec:Steiner-2cut-equiv}), the result of this subsection directly generalizes \Cref{lemdepth}.

\subsection{Finding an optimal \texorpdfstring{$k$-cut}{k-cut} tree}\label{sec:k-cut-opt-find}

Let $T'$ be a $k$-cut STT on $S$ that serves the search sequence $X$ of length $m$ with minimal cost among all $k$-cut STTs. Let $r$ be the root of $T'$, let $T'_1, \dots, T'_t$ be the subtrees rooted at children of $r$, and let $X_i$ be the subsequence of $X$ consisting of searches to nodes of $T'_i$. Then, $\cost_{T'}(X)=m + c_1 + \cdots + c_t$, where $c_i = \cost_{{T}'_i}(X_i)$, for all $i$.

By definition, the trees $T'_i$ are $k$-cut trees. Our strategy is to find $r$ and to recursively compute $k$-cut STTs $T_1, \dots, T_t$ on the components of $S \setminus r$ that achieve cost at most $c_i$ for their respective sequences $X_i$, i.e.\ for the relevant frequencies $p:V(T'_i) \rightarrow \mathbb{N}$, for all $i$. We then return the tree $T$ obtained by letting the roots of $T_1, \dots, T_t$ be children of the root $r$.

The cost equation translates into the straightforward dynamic program~\Cref{alg:dp}. We call those sets of nodes $A \subseteq V(S)$ \emph{$k$-admissible} for which $S[A]$ is connected and has boundary size at most $k$. We call a node $x \in A$ a \emph{$k$-admissible root} of $A$ if the node sets of all connected components of $S[A] \setminus x$ are $k$-admissible.

Procedure OPT-STT in \Cref{alg:dp} computes an optimal $k$-cut tree for a $k$-admissible set $A \subseteq V(S)$ and the relevant frequencies $p: A \rightarrow \mathbb{N}$. The initial call is OPT-STT$(V(S))$.
Only the root and the total cost are returned, the full tree  can be reconstructed by collecting the roots from the recursive calls with standard bookkeeping. 

Line~\ref{optln2}--\ref{optln3} is the base case (a tree of a single node). In Line~\ref{alg-line:for-adm-root} the $k$-admissible root of the current subset is selected, and in Line~\ref{optln5}--\ref{optln7} the optimal subtrees are found. In Line~\ref{optln8} the total cost is computed with the chosen root and the roots of the optimal subtrees as its children. The first term counts the number of times the root is accessed, and the second term adds the costs of accesses in the subtrees. The correctness of the algorithm follows from the preceding discussion.

\paragraph{Preprocessing.}

The dynamic program is over all nonempty $k$-admissible subsets of $V(S)$. In a preprocessing step we enumerate all such sets, indexed by their cuts. As only cuts of size at most $k$ are relevant, we can iterate through them by traversing $S$ with $k$ pointers. For each cut, we do another traversal of the tree, enumerating the set of nodes in the corresponding $k$-admissible subset. Observe that some cuts lead trivially to an empty set of nodes (when cut-edges point away from each other), and some cuts contain redundant edges. We can easily detect and remove these cases.

\begin{algorithm}[t]
	\caption{Finding the optimal $k$-cut STT on $A$}\label{alg:dp}
	\begin{algorithmic}[1]
		\Statex \textbf{Input:} $k$-admissible set of nodes $A \subseteq V(S)$, search frequencies $p: A \rightarrow \mathbb{N}$.
		\Statex \textbf{Output:} $(r,c)$, for optimal $k$-cut STT on $A$ with root $r$ and cost $c$. 
		\vspace{0.05in}
		
		\Procedure{OPT-STT}{$A$}
			\If{$A = \{r\}$} \hfill$\triangleright$ single node \label{optln2}
				\State \textbf{return} $(r, p(r))$ \label{optln3}
			\EndIf
			\For {$k$-admissible root  $r \in A$}\label{alg-line:for-adm-root}  %
				\State $A_1, \dots, A_t \leftarrow$ node sets of the connected components of $S[A] \setminus r$.\label{optln5}
				\For {$i = 1,\dots, t$}\label{optln6}
					\State $(r_i,c_i) \leftarrow \Call{OPT-STT}{A_i}$ \hfill$\triangleright$ optimal subtrees $T_1, \dots, T_t$\label{optln7}
				\EndFor
				\State Let $C_r = \displaystyle\sum_{x \in A} {p(x)} + \displaystyle\sum_{i \in [t]} {c_i} $. \hfill$\triangleright$ cost of tree with root $r$\label{optln8}
			\EndFor
			\State \textbf{return} $(r,C_r)$ for $r$ that minimizes $C_r$. \label{optln9}
		\EndProcedure
	\end{algorithmic}
\end{algorithm}

\paragraph{Admissible roots.}

We now discuss the finding of $k$-admissible roots (\Cref{alg-line:for-adm-root}).

\begin{lemma}\label{lem:admr}
	Let $A \subseteq V(S)$ be a $k$-admissible set and assume $k \ge 2$.
	\begin{compactenum}[(i)]
		\item If $|\delta(A)| < k$, then every node $r \in A$ is a $k$-admissible root of $A$.\label{item:admr-each}
		\item If $|\delta(A)| = k$, then the set of $k$-admissible roots of $A$ is $\ch(\delta(A)) \cap A$.\label{item:admr-ch}
	\end{compactenum}
\end{lemma}
\begin{proof}
	\emph{(i)} Let $|\delta(A)| < k$, let $r \in A$, and let $C$ be a connected component of $S[A] \setminus r$. Then, by \Cref{p:comp-bounds}, $\delta(C) \subseteq \delta(A) \cup \{r\}$, and thus $|\delta(C)|  \le |\delta(A)| + 1 \le k$. Thus, $r$ is a $k$-admissible root of $A$.
	
	\emph{(ii)} Let $|\delta(A)| = k$. %

	Let $r \in A$ be a $k$-admissible root. If two or more boundary nodes of $A$ are neighbors of $r$, then $r$ is in $\ch(\delta(A))$. Otherwise, at least one boundary node $u \in \delta(A)$ is not a neighbor of $r$, so by \Cref{p:comp-bounds}, there must be some connected component $C$ of $S[A] \setminus r$, such that $u \in \delta(C)$. Then, by assumption, there is some $v \in \delta(A) \setminus \delta(C)$ (otherwise, by \Cref{p:comp-bounds}, $V(C)$ has boundary size $k+1$ and $r$ is not $k$-admissible). Node $v$ is either a neighbor of $r$ or a boundary node of a component $C' \neq C$ of $S[A] \setminus r$. The path between $u$ and $v$ must pass through $r$, implying $r \in \ch(\delta(A))$.
	
	Conversely, assume $r \in \ch(\delta(A)) \cap A$. Let $u, v \in \delta(A)$ be two distinct nodes. Then $r$ is on the path between $u$ and $v$. Thus, $u$ and $v$ are in different connected components of $S[A] \setminus r$. Now \Cref{p:comp-bounds} implies that the boundary of each connected component $C$ of $S[A] \setminus r$ is a proper subset of $\delta(A) \cup \{r\}$, and thus has boundary size at most $k$.
\qedd
\end{proof}

Given \Cref{lem:admr}, the enumeration of $k$-admissible roots in Line $4$ is straightforward, via a traversal of the subtree $S[A]$. In case (\ref{item:admr-each}) we traverse the entire tree $S[A]$, in case (\ref{item:admr-ch}) we traverse the tree of paths from some boundary node $v \in \delta(A)$ to the other boundary nodes $\delta(A) \setminus \{v\}$ (found e.g.\ with a breadth-first search). The cuts of the components can be found in linear time by straightforward data structuring.

\paragraph{Running time.}
In the preprocessing stage we enumerate $\fO(n^k)$ cuts, and for each cut we do a linear-time traversal to find the corresponding $k$-admissible set, all within time $\fO(n^{k+1})$.

The recursive calls of OPT-STT are for smaller $k$-admissible sets. %
Therefore, during the preprocessing phase we sort the $k$-admissible sets by size, and in the dynamic programming table we fill in the entries by increasing order of size. It remains to show that filling in one entry takes time $\fO(n)$, from which the overall running time of $\fO(n^{k+1})$ follows.

Lines \ref{optln2}--\ref{optln3} take $\fO(1)$ time. In Line~\ref{alg-line:for-adm-root} we iterate over all $k$-admissible roots, which, by the preceding discussion, takes time $\fO(n)$. In Line~\ref{optln5} we read out the connected components indexed by their cuts, computed during preprocessing. Line~\ref{optln8} takes $\fO(1)$ time, as the first term can be precomputed for all $k$-admissible sets, and the second term is collected from the recursive calls.

Line~\ref{optln7} is nested in two loops (iterating through possible root nodes, and through each connected component after the removal of the root). Nonetheless, it is easy to see that it is executed at most twice for each \emph{edge} in $S[A]$ (once for each orientation). The total number of recursive calls is therefore at most $2n-2 \in \fO(n)$, as is the cost of taking the minimum in Line~\ref{optln9}. 

 Using Lemma~\ref{p:depth_approx-improved} and setting $k=2t$, we obtain:

\restatethma*

\paragraph{Remark.} It is tempting to try extending the approximation algorithm (with some ratio $r>2$) to the the easiest $k=1$ case, i.e.\ when the STT is a rooted version of $S$. Unfortunately, $1$-cut trees cannot give an $o(n/\log{n})$-approximation of the STT optimum. To see this, take $S$ to be a path, and observe that every rooted version of $S$ has average depth $\Omega(n)$, whereas a BST on $S$ (which is, in particular, a $2$-cut tree) has maximum depth $O(\log{n})$.

\section{Rotations in \texorpdfstring{$k$}{k}-cut trees}\label{sec4}

As discussed in \S\,\ref{sec:intro}, an essential feature of the classical BST model is that the rotation-distance between two trees of size $n$ is $\fO(n)$.\footnote{The number of rotations needed to transform one $n$-node BST into another is at most $2n-6$ and there are pairs of trees requiring this many rotations for all $n>10$~\cite{STT88,Pournin}. As BSTs are trivially Steiner-closed (i.e.\ $2$-cut), a similar lower bound also holds for Steiner-closed STTs.} In particular, if we only do rotations on the search path, as in most natural algorithms, then the cost of rotations can be charged to the cost of searching, i.e.\ of simply traversing the search path. In STTs the situation is different, as there are pairs of trees of size $n$ that are $\Omega(n \log{n})$ rotations apart~\cite{Cardinal18}. Nonetheless, for $k$-cut STTs we show the following. %

\restatethmb*

In our study of adaptive STTs (i.e.\ our generalized Splay) in \S\,\ref{sec5} and \S\,\ref{sec6} we will only make use of the case $k=2$, Theorem~\ref{thm3} is thus more general than strictly necessary for our immediate purposes. STTs are in bijection with vertices of \emph{tree associahedra}, whose edges correspond to rotations. In terms of this structure, Theorem~\ref{thm3} can be interpreted as follows. While the skeleton of a tree associahedron (for a tree of size $n$) can have diameter $\Omega{(n\log{n})}$, its vertices corresponding to $k$-cut STTs induce a connected subgraph of diameter $O(kn)$. %

Before proceeding with the proof, we make an observation that will be useful later.
\begin{observation}\label{p:rot-destroys-kcut}
	Let $T$ be an STT, and let $T'$ be the STT that we obtain by rotating at node $q$ with parent $p$. Then, for each $x \in V(T) \setminus \{p,q\}$, we have $V(T'_x) = V(T_x)$ and thus $|\delta(T'_x)| = |\delta(T_x)|$.
\end{observation}

In particular, if $T$ is a $k$-cut STT, then we only have to consider the boundary $\delta(T'_p)$ to see if $T'$ is still a $k$-cut STT, since $|\delta(T'_q)| = |\delta(T_p)| \le k$.

\medskip

We prove \Cref{thm3} by induction. For $k=1$, we directly rotate $T$ into $T'$. For $k > 1$, we reduce $k$ by
rotating both $T$ and $T'$ into $(k-1)$-cut trees. By reversing the rotation sequence that we used to transform $T'$, we obtain a sequence of rotations from $T$ to $T'$.

We first focus on the number of rotations and ignore the number of pointer moves needed; in \S\,\ref{sec:rot-impl} we present a refined algorithm. The following lemma concerns the case $k=1$. Recall that $1$-cut trees are precisely those that can be obtained by rooting $S$ at some node $r$. Let $S^r$ denote this rooted tree.

\begin{lemma}\label{p:1cut-dist}
	There is a sequence of at most $n-1$ rotations that transforms $S^{r_1}$ into $S^{r_2}$, for arbitrary nodes $r_1, r_2 \in V(S)$. All intermediate trees are $1$-cut trees.\label{p:diam-base}
\end{lemma}
\begin{proof}
	Let $P = (x_1, \dots, x_t)$ be the search path of $r_2$ in $S^{r_1}$, with $x_1 = r_1$ and $x_t = r_2$. We rotate at the nodes $x_2, \dots, x_t$ (in this order).
	
	As $t \leq n$, we clearly make at most $n-1$ rotations. We show inductively that after rotating at $x_i$, for all $i=2,\dots,t$, the obtained tree is $S^{x_i}$. The claim follows, as the last rotation is at $x_t = r_2$.
	
	Consider the tree after rotating at $x_{i-1}$. By the inductive claim, $x_{i-1}$ is the root, and since $\{x_{i-1},x_i\}$ is an edge of $S$, node $x_i$ is the child of the root.
	The next rotation brings $x_i$ to the root, making it the parent of $x_{i-1}$. All other nodes whose parent changes must be in the subtree of $S$ delimited by $x_i$ and $x_{i-1}$, but since $x_{i-1}$ and $x_i$ are connected by an edge in $S$, there are no such nodes. Thus, the edge-set of the tree remains the same and equals the edge set of $S$. Since all intermediate trees are of the form $S^{x_i}$, they are $1$-cut. %
	\qedd
\end{proof}

We continue with the transformation of $k$-cut STTs into $(k-1)$-cut STTs. The following lemma shows that certain rotations strictly ``improve'' the tree, by reducing the boundary size of some subtree and not affecting the boundary sizes of other subtrees.

\begin{figure*}[h]
  \centering
  \includegraphics[width=14cm]{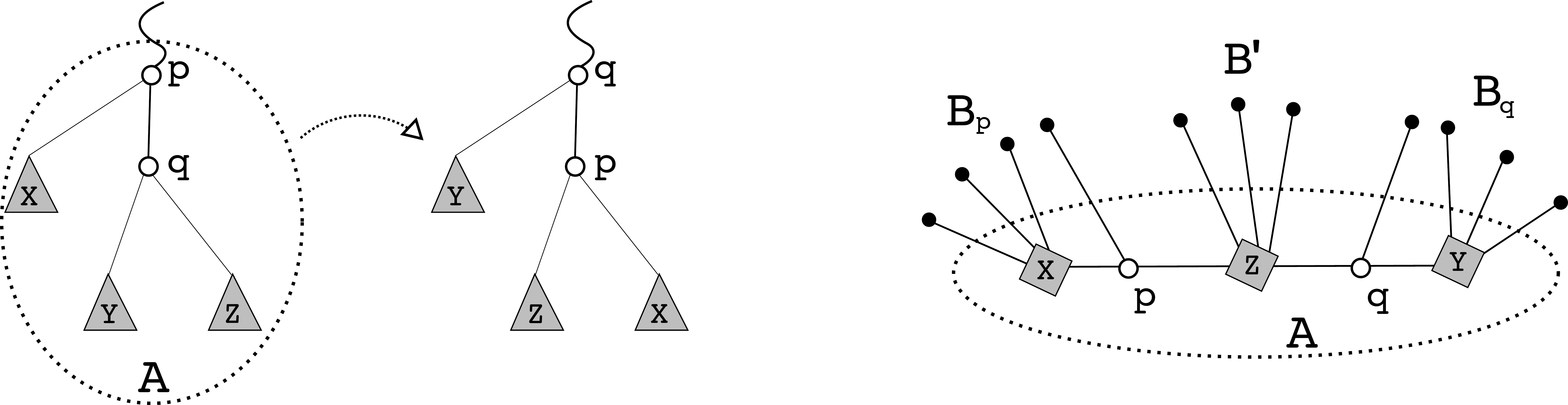}
  \caption{Illustration of the proof of Lemma~\ref{p:rot-step}. \label{fig8}}
\end{figure*}

\begin{lemma}\label{p:rot-step}
	Let $T$ be a $k$-cut STT on $S$, where $k \ge 2$. Let $q$ be a node of $T$ with parent $p$ such that $|\delta(T_q)| = k$ and $|\delta(T_p)| = k-1$. Then, rotating at $q$ produces a $k$-cut STT $T'$, where $|\delta(T'_q)| \le k-1$ and for each node $x \in V(T) \setminus \{q\}$, we have $|\delta(T'_x)| \le |\delta(T_x)|$.
\end{lemma}
\begin{proof}
	Let $A = V(T_p)$ and let $B = \delta(T_p)$. The nodes $p$ and $q$ split $B$ into three parts:
	\begin{itemize}
		\itemsep0mm
		\item $B_p$ contains all nodes in $B$ that $p$ separates from $q$,
		\item $B_q$ contains all nodes in $B$ that that $q$ separates from $p$,
		\item $B'$ contains all other nodes in $B$.
	\end{itemize}
	
	Clearly, $B$ is the disjoint union of $B_p$, $B_q$, $B'$. Observe that $\delta(T_q) = B' \cup B_q \cup \{p\}$.

	We first claim that $B_q \neq \emptyset$. Suppose otherwise. Since $T$ is a $k$-cut STT, $q$ is a $k$-admissible root of $V(T_q)$. Thus, by \Cref{lem:admr}, we have $q \in V(\ch(\delta(T_q))) = V(\ch(B' \cup B_q \cup \{p\})) = V(\ch(B' \cup \{p\}))$. But $q$ does not separate any two vertices in $B' \cup \{p\}$, nor is it contained in $B' \cup \{p\}$, a contradiction.
	
	We conclude the proof by showing $|\delta(T'_p)|, |\delta(T'_q)| \le k-1$. For all other nodes $x$, we have $|\delta(T'_x)| \le |\delta(T_x)|$ by \Cref{p:rot-destroys-kcut}. (See Figure~\ref{fig8}.)
	\begin{itemize}
		\itemsep0mm
		\item For $p$, we have $\delta(T'_p) = B_p \cup B' \cup \{q\}$. Since $B_q \neq \emptyset$, we have $|B_p \cup B'| \le |B| - 1 = k-2$. Thus, $|\delta(T'_p)| \le k-1$.
		\item For $q$, we have $|\delta(T'_q)| = |B| = |\delta(T_p)| = k-1$.~\qedd
	\end{itemize} 
	\let\qed\relax
\end{proof}

It is easy to see that in a $k$-cut STT that is not a $(k-1)$-cut STT, we always have a node $q$ that satisfies the requirements of \Cref{p:rot-step}. Thus, we can repeatedly apply \Cref{p:rot-step} and finish after at most $n$ steps. In fact, we can do slightly better:

\begin{lemma}\label{p:kcut-decr-k}
	Let $T$ be a $k$-cut STT on $S$, with $k \ge 2$. Then $T$ can be transformed into a $(k-1)$-cut STT with at most $n-k$ rotations, so that each intermediate tree is a $k$-cut tree.
\end{lemma}
\begin{proof}
	Suppose $T$ is not a $(k-1)$-cut STT. Since $|\delta(T)| = 0$, there is a node $q$ with parent $p$ such that $|\delta(T_q)| = k$ and $|\delta(T_p)| < k$. Since $\delta(T_q) \subseteq \delta(T_p) \cup \{p\}$, we have $|\delta(T_p)| = k-1$. We can thus apply \Cref{p:rot-step}, increasing the number of nodes $x$ with $|\delta(T_x)| \le k-1$ by at least one. Repeat this step until we have a $k$-cut STT. In the original tree $T$, every node $x$ of depth at most $k$ already satisfies $|\delta(T_x)| \le \depth(x) - 1 \le k-1$. As such, we can ignore these nodes and need at most $n-k$ steps.\qedd
\end{proof}

The first part of \Cref{thm3} now follows by induction as outlined above:
\begin{lemma}\label{p:rot-dist-noimpl}
	Let $T, T'$ be $k$-cut STTs. Then we can transform $T$ into $T'$ with $(2k-1)n - (k+1)k + 1$ rotations.
\end{lemma}
\begin{proof}
	If $k=1$, we need $n-1$ rotations, by \Cref{p:1cut-dist}. Otherwise, we transform the two $k$-cut STTs $T, T'$ into $(k-1)$-cut STTs $T'', T'''$, using $2(n-k)$ rotations, by \Cref{p:kcut-decr-k}. By induction, the rotation distance between $T''$ and $T'''$ is at most $(2(k-1)-1)n - k(k-1) + 1$. The rotation distance between $T$ and $T'$ is therefore at most
	\begin{align*}
		& 2(n-k) + (2(k-1)-1)n - k(k-1) + 1\\
		& = (2k-1)n - 2k - k(k-1) + 1\\
		& = (2k-1)n - (k+1)k + 1. ~~\qedd%
	\end{align*}
	\let\qed\relax
\end{proof}

\subsection{Implementation}\label{sec:rot-impl}

We now finish the proof of Theorem~\ref{thm3}, by showing that the constructed rotation-sequence can be carried out with $\fO(k^2n)$ pointer-moves and oracle calls. (This is necessary if we are to apply it in the STT model.)

Implementing \Cref{p:1cut-dist} requires no more than the search for $r_2$ in $S^{r_1}$ and $n-1$ rotations, which altogether needs $\fO(n)$ steps.
It remains to show that %
\Cref{p:kcut-decr-k} can be implemented in $\fO(kn)$ steps, from which the $O(k^2n)$ bound of Theorem~\ref{thm3} will follow. %

\paragraph{The boundary-size data structure.}
To detect nodes at which we can rotate (as in \Cref{p:rot-step} and \Cref{p:kcut-decr-k}), we need to efficiently compute the boundary sizes of subtrees rooted at arbitrary nodes. Furthermore, these boundary sizes need to be updated during rotations. We show that a data structure that stores all boundary sizes can be obtained in an $O(k^2n)$-time preprocessing, with $O(k)$-time updates for every rotation. As we do $O(kn)$ rotations altogether, these costs can be absorbed in the $O(k^2n)$ bound of Theorem~\ref{thm3}.

Consider a node $q$ with parent $p$ in $T$, and suppose that $\delta(T_p)$ is known. We know that $p \in \delta(T_q)$, since $V(T_q)$ is a connected component of $S[V(T_p)] \setminus p$. Further, $\delta(T_q) \setminus \{p\} \subseteq \delta(T_p)$. Each node $x \in \delta(T_p)$ is a boundary node of $T_q$ if and only if $x$ and $q$ lie on the same side of $p$ in $S$. Thus, we can compute $\delta(T_q)$ from $\delta(T_p)$ using $|\delta(T_p)|$ oracle calls. It follows that, if $T$ is a $k$-cut tree, we can compute $\delta(T_x)$ for all $x\in V(T)$ in an $O(n)$-time traversal of $T$ starting from the root, with $O(nk)$ oracle-calls.

Consider now a rotation of the edge $\{q,p\}$ that transforms $T$ into $T'$. %
We know that $\delta(T'_q) = \delta(T_p)$, and we can compute $\delta(T'_p)$ as described above. For all other $x \in V(T) \setminus \{p,q\}$, we have $\delta(T'_x) = \delta(T_x)$. It follows that the boundary-size data structure can be updated with $O(k)$ oracle calls.

\paragraph{The algorithm.} We proceed via a modified depth-first search, starting with a pointer at the root. Mark nodes as \texttt{unvisited}, \texttt{visited}, or \texttt{finished}, with all nodes initially \texttt{unvisited}.

Let $T$ denote the current tree, and let $y$ be the node at the pointer. If $y$ is \texttt{unvisited}, we mark it \texttt{visited}. If all children of $y$ are \texttt{finished}, we mark $y$ as \texttt{finished} and move the pointer to the parent of $y$, unless $y$ is the root, in which case we are done.

Otherwise, if $y$ has an \texttt{unvisited} child, we pick one such child $x$. If $|\delta(T_x)| \le k-1$, then we move the pointer to $x$. Otherwise, we rotate the edge $\{y,x\}$, mark $x$ as \texttt{visited}, and keep the pointer at $y$.

\medskip

By \Cref{p:kcut-decr-k}, the algorithm uses at most $n-k$ rotations. Boundary-size queries are handled by the previously described data structure. %
It remains to show that the algorithm terminates with $O(n)$ pointer moves and produces a $(k-1)$-cut tree. Towards this claim, observe the following invariants, easily shown by induction:
\begin{compactenum}[(i)]
	\item if a node is \texttt{finished}, resp.\ \texttt{unvisited} then all its descendants in $T$ are \texttt{finished}, resp.\ \texttt{unvisited},
	\item if the node at the pointer is \texttt{visited} then all its children are either \texttt{finished} or \texttt{unvisited},
	\item every node $x$ marked \texttt{visited} or \texttt{finished} satisfies $|\delta(T_x)| \le k-1$.
\end{compactenum}

\medskip

Invariants (i)--(ii) imply that the algorithm does not get stuck, and since every pointer move changes a marking from \texttt{unvisited} to \texttt{visited} or from \texttt{visited} to \texttt{finished}, after $2n$ steps all nodes are \texttt{finished}, with the pointer at the root. 

Invariant (iii) implies that all rotations performed are of the kind defined in \Cref{p:rot-step}, and that the algorithm produces a $(k-1)$-cut tree.

\subsection{Steiner-closed trees}\label{sec:Steiner-2cut-equiv}

We now prove the equivalence between $2$-cut STTs and \emph{Steiner-closed STTs}, defined in \S\,\ref{sec2}. %

\begin{lemma}\label{p:Steiner-2cut-equiv}
	A search tree $T$ on $S$ is Steiner-closed if and only if it is a $2$-cut STT.
\end{lemma}
\begin{proof}
	Suppose that there is a node $x \in V(T)$ such that $|\delta(T_x)| \ge 3$. Since $S[V(T_x)]$ is connected, every pair of nodes from $\delta(T_x)$ is connected through a path that lies entirely in $V(T_x)$, except for its endpoints. Thus $V(\ch(\delta(T_x))) \setminus \delta(T_x) \subseteq V(T_x)$. Moreover, $\ch(\delta(T_x))$ is a tree with leaves $\delta(T_x)$ (by \Cref{p:delta-ch-leaves}). As there are more than two leaves, there must be an inner node $v \in V(T_x)$ with degree at least $3$ in $\ch(\delta(T_x))$. Let $P$ be the search path of the parent of $x$ in $T$. Then $v \notin P$, but $v$ has degree at least $3$ in $\ch(P)$, so $P$ is not Steiner-closed.
	
	Suppose $T$ is not Steiner-closed. Then there is a path with node set $P$ from the root to some node $x$ such that $P$ is not Steiner-closed. Let $y \in \ch(P) \setminus P$ be a node with degree at least $3$ in $\ch(P)$. There must be three distinct nodes $u, v, w \in P$ that are pairwise separated by $y$. W.l.o.g., no other node in $P$ separates either of $u,v,w$ from $y$ (i.e., choose $u,v,w$ such that their distance to $y$ is minimal). Clearly, $y \in V(T_x)$. Let $z$ be the child of $x$ such that $y \in V(T_z)$. Then $\{u,v,w\} \subseteq \delta(T_z)$, so $T$ is not a $2$-cut STT. \qedd
\end{proof}

We have the following corollary of \Cref{thm3}.

\begin{corollary}
	Given two Steiner-closed STTs $T$ and $T'$ on the same $n$-node tree $S$, we can transform $T$ into $T'$ through a sequence of at most $3n - 5$ rotations, such that all intermediate trees are Steiner-closed. Moreover, starting with a pointer at the root, we can transform $T$ into $T'$ through a sequence of $\fO(n)$ pointer moves and rotations at the pointer.
\end{corollary}

\section{Splay trees on trees}\label{sec5}

In this section we extend Splay trees, introduced by Sleator and Tarjan~\cite{ST85} to the setting of trees on trees. The Splay tree is an adaptive BST, re-arranged via rotations after every search. (Other operations such as insert, delete, split, etc.\ are also defined, but for the sake of simplicity we only focus on searches.) 

A search for an element $x$ proceeds as in a normal BST, following the search path from the root to node $x$. Afterwards, $x$ is rotated to the root in a series of local steps (called by Sleator and Tarjan the ZIG, ZIG-ZIG, and ZIG-ZAG steps), with the entire transformation of the tree called \emph{splaying} $x$. Intuitively, the effect of splaying, besides bringing $x$ to the root, is to approximately halve the depth of every node on the search path~\cite{Subramanian96, ESA15}. 

Splaying can be defined in a number of equivalent ways. A simple view is that $x$ is rotated to the root, with the rotations grouped in consecutive pairs. For each pair, if the parent $p$ and the grandparent $g$ of $x$ are \emph{on the same side of $x$} (i.e.\ both to the right or both to the left), then we \emph{skip ahead} and rotate at $p$ first, then at $x$ (this is the ZIG-ZIG step). Otherwise, we do both rotations at $x$, without skipping ahead (this is the ZIG-ZAG step). If the search path is of odd length, then the entire process is finished by a simple rotation (a ZIG step). The ``skipping ahead'' of ZIG-ZIG is crucial for the efficiency of Splay; a simple rotate-to-root strategy is well-known to be inefficient on some examples. 

The above description of Splay trees can be easily extended to STTs, with an appropriate generalization of the  parent and grandparent being ``on the same side of $x$''.
We now describe this generalization. %
To splay $x$, we repeatedly apply one of three types of operations (the extensions of ZIG, ZIG-ZIG, and ZIG-ZAG). For reasons that will become clear later, we define splaying more generally, bringing $x$ not necessarily to the root, but to a given point on the search path. We refer to Algorithm~\ref{alg:splay} and Figure~\ref{fig4} for details.

 \begin{algorithm}
  \caption{SplayTT (generalized splaying procedure for STTs)}\label{alg:splay}
  \begin{algorithmic}[1]
    \Statex \textbf{Input:} search tree $T$ on $S$, node $x$ to be splayed until node $y$ is parent of $x$ in $T$.
    \Procedure{\textsc{splay}$(x,y)$}{}
    \While{$x.parent \neq y$}
	\State $p \leftarrow x.parent$
	\If{$p.parent = y$}\hfill$\triangleright$ (ZIG)
		\State rotate at $x$ 
	\Else
		\State $g \leftarrow p.parent$
		\If{$p$ separates $x$ and $g$ in $S$}\hfill$\triangleright$ (ZIG-ZIG)\label{sln8}
			\State rotate at $p$, then rotate at $x$ 
		\ElsIf{$x$ separates $p$ and $g$ in $S$}\hfill$\triangleright$ (ZIG-ZAG)\label{sln10}
			\State rotate twice at $x$
		\EndIf
	\EndIf	        
	\EndWhile

   \EndProcedure
    \end{algorithmic}
\end{algorithm}

Observe that splay$(x,y)$ includes as special case the operation of splaying $x$ all the way to the root, when called as splay$(x,null)$. It is easy to verify that when the underlying tree $S$ is a path, i.e.\ in the BST case, the defined splaying operation is identical with the classical splaying of Sleator and Tarjan. %

The algorithm, as described, is not well-defined for arbitrary STTs. Besides the two cases (Line~\ref{sln8}) $p$ separates $x$ and $g$, and (Line~\ref{sln10}) $x$ separates $p$ and $g$, there is a third possible case, when none of $x$, $p$, and $g$ separate the other two. (A hypothetical fourth case, where $g$ separates $p$ and $x$ cannot arise, as in that case $g,p,x$ would not appear in this order on the search path.)

It is easy to see that the non-separating third case can only arise if the search tree is not Steiner-closed. Therefore, we define and analyse splaying under the restriction that the search tree $T$ is Steiner-closed at all times, including in its initial state. In this way Algorithm~\ref{alg:splay} is well-defined, and amenable for analysis. A remaining technical challenge is that even if $T$ is Steiner-closed before splaying, the act of splaying may destroy this property. (We show a simple example in Figure~\ref{fig7}.)

Let $P = P_x$ be the node set of the search path of $x$ in a Steiner-closed tree $T$, and let $C = \ch(P)$. The difficulty, intuitively, is that $C$ is not a path as in the BST case, but a tree. Call a node of $C$ of degree more than $2$ a \emph{branching node}. We prove an alternative characterization of branching nodes that will be useful later.

\begin{lemma}\label{p:branch-equiv}
	Let $T,P$ be defined as above. A node $p \in P$ is a branching node of $P$ if and only if $|\delta(T_p)| = 2$ and $p$ has a child $q$ on $P$ with $|\delta(T_q)| = 1$.
\end{lemma}
\begin{proof}
	Let $p \in P$ be a branching node. Then there are three nodes $y_1, y_2, y_3 \in P$ that are pairwise separated by $p$. W.l.o.g., no node in $P$ separates $p$ from either of the nodes $y_1, y_2, y_3$. By the search tree property, at most one of $y_1, y_2, y_3$ can be a descendant of $p$. If $y_1,y_2,y_3$ are all ancestors of $p$, then the search path to the parent of $p$ is not Steiner-closed. As such, one of the nodes, say $y_3$, is a descendant of $p$, and $y_1,y_2$ are ancestors of $p$, which implies $\delta(T_p) = \{y_1, y_2\}$. (Since $T$ is Steiner-closed and thus a $2$-cut tree, $|\delta(T_p)| = 2$.)
	
	Moreover, since $p$ has a descendant $y_3$ on $P$, it also has a child $q$ on $P$, which must be in the same connected component of $S \setminus p$ as $y_3$. We have $p \in \delta(T_q)$ and $\delta(T_q) \setminus \{p\} \subseteq \delta(T_p) = \{y_1, y_2\}$. But since $p$ separates $q$ from both $y_1$ and $y_2$, we have $\delta(T_q) = \{p\}$.
	
	For the other direction, suppose $q$ is the child of $p$ and $|\delta(T_p)| = 2$, $|\delta(T_q)| = 1$. Clearly, $\delta(T_q) = \{p\}$. Say $\delta(T_p) = \{y_1, y_2\}$. Since $y_1, y_2 \notin \delta(T_q)$, the three nodes $y_1, y_2, q$ must be pairwise separated by $p$. This means that $p$ has degree $3$ in $\ch(\{y_1, y_2, q\})$, and thus degree at least $3$ in $\ch(P)$. \qedd
\end{proof}

From Lemma~\ref{p:branch-equiv} it follows, in particular, that neither the root $r$ of $T$, nor the searched node $x$ can be branching (since $|\delta(T_r)| = 0$, and $x$ has no child on $P$). The following lemma shows that we can only lose the Steiner-closed property by rotating the edge between a branching node and its relevant child.

\begin{lemma}\label{p:branching-destr-Steiner}
	Let $T$ be a Steiner-closed STT, let $q \in V(T)$ with parent $p \in V(T)$. Let $T'$ be the tree obtained by rotating the edge $\{q,p\}$. Then $T'$ is Steiner-closed if and only if $|\delta(T_p)| \neq 2$ or $|\delta(T_q)| \neq 1$.
\end{lemma}
\begin{proof}
	Suppose that $|\delta(T_p)| = 2$ and $|\delta(T_q)| = 1$, i.e.\ $p$ is a branching node of the search path of $q$. Then $p$ has two ancestors $y_1, y_2$ such that $p$ has degree $3$ in $\ch(\{y_1, y_2, q\})$. The search path $P'_q$ of $q$ in $T'$ contains $y_1, y_2, q$, but not $p$, so $P'_q$ is not Steiner-closed.
	
	Conversely, suppose that $T'$ is not Steiner-closed, so $|\delta(T'_p)| > 2$ by \Cref{p:rot-destroys-kcut}. Since $\delta(T'_p) \subseteq \delta(T'_q) \cup \{q\}$, and $|\delta(T'_q)| \le 2$ by \Cref{p:rot-destroys-kcut}, we have $|\delta(T'_q)| = 2$, say $\delta(T'_q) = \{y_1, y_2\}$ and $\delta(T'_p) = \{y_1, y_2, q\}$. We already have $|\delta(T_p)| = |\delta(T'_q)| = 2$. It remains to show that $|\delta(T_q)| = 1$. We claim that $\delta(T_q) = \{p\}$.
	
	Suppose, to the contrary, that $y_i \in \delta(T_q)$ for some $i \in \{1,2\}$. Then $q$ cannot separate $y_i$ from $p$, since $y_1 \in \delta(T_p)$. Moreover, no ancestor of $p, q$ is in $\ch(\{p,q\})$. In particular, $\ch(\{p,q,y_i\})$ has a node $z$ of degree $3$, which is not an ancestor of $p,q$. As such, the search path of $q$ in $T$ is not Steiner-closed, a contradiction. \qedd
\end{proof}

Our goal is to remove all branching nodes from the search path $P_x$ before splaying $x$ to the root. Note that by \Cref{p:branch-equiv}, a branching node can never be the child of another branching node. Our strategy is to splay up a branching node $b$ to become the child of its lowest branching ancestor $b'$, or to become the root, if $b$ has no branching ancestors. In both cases $b$ loses the branching property. It is crucial that during this process no new branching nodes are created on $P_x$ and $b'$ (if exists) stays a branching node. We show this in the following.

\begin{lemma}\label{p:branch-rot-unaff}
	Let $T$ be an STT, let $P_x$ be the search path of $x$ in $T$, and let $q \in P$ be a node with non-branching parent $p \in P$. Let $T'$ be the tree obtained by rotating the edge $\{q,p\}$, and let $P'_x$ be the search path of $x$ in $T$. Then
	\begin{compactenum}[(i)]
		\item Every branching node of $P'_x$ is a branching node of $P_x$;
		\item Every branching node of $P_x$, except possibly $q$, is a branching node of $P'_x$.
	\end{compactenum}
\end{lemma}
\begin{proof}
	(i) Since $V(P'_x) \subseteq V(P_x)$, no non-branching node on the search path becomes branching and no new branching node enters the search path of $x$.
	
	(ii) Let $b \neq q$ be a branching node of $P_x$, with child $c$ on $P_x$. The branching node $b$ is not affected by the rotation, so $|\delta(T'_q)| = |\delta(T_q)| = 2$ by \Cref{p:rot-destroys-kcut}. If $c \neq q$, then $c$ is not affected by the rotation either, so $|\delta(T'_c)| = |\delta(T_c)| = 1$. Otherwise, if $c = q$, then $|\delta(T'_c)| = |\delta(T_q)| = 1$. Either way, $b$ stays a branching node. \qedd
\end{proof}

We are now ready to describe the search operation in our generalized SplayTT structure (Algorithm~\ref{alg:splay2}). Like Splay, it starts with a normal search to the target node $x$ in the (Steiner-closed) search tree $T$, identifying the search path of $x$ (Line~\ref{spll2}). Then, in a first phase of splaying (Lines~\ref{spll3}--\ref{spll6}), it transforms the search path, such as to remove all branching nodes. In Line~\ref{spll3} the branching nodes $x_1, \dots, x_k$ on the search path are found, indexed from the root towards $x$. Then, in Lines~\ref{spll5}--\ref{spll6}, from bottom-to-top, each branching node $x_i$ is splayed up, until it becomes the child of its nearest branching ancestor $x_{i-1}$ (or it becomes the root, in case of $x_1$). By \Cref{p:branch-rot-unaff}, splay$(x_i, x_{i-1})$ creates no new branching nodes on the search path of $x$, and afterwards, $x_i$ is no longer a branching node. As such, the search path of $x$ after the first phase contains no branching nodes, and by Lemma~\ref{p:branching-destr-Steiner}, the tree remains Steiner-closed.
In the second phase (Line~\ref{spll7}), node $x$ is splayed all the way to the root on the transformed search path. Again, by Lemma~\ref{p:branching-destr-Steiner}, the tree remains Steiner-closed.

 \begin{algorithm}
  \caption{SplayTT search operation}\label{alg:splay2}
  \begin{algorithmic}[1]
    \Statex \textbf{Input:} Search tree $T$ on $S$. Node $x$ to be searched. 
    \Procedure{\textsc{search}$(x)$}{}
    \State Follow search path in $T$ from root to $x$.\label{spll2}
    \State Identify branching nodes $x_1, \dots, x_k$ on the search path of $x$.\hfill$\triangleright$ (phase 1)\label{spll3}
    \State Let $x_0 = null$.\label{spll4}
    \For{$i = k, \dots, 1$ }\label{spll5}
	\State splay$(x_i, x_{i-1})$\label{spll6}
	\EndFor

	\State splay$(x, null)$\hfill$\triangleright$ (phase 2)\label{spll7}
   \EndProcedure
   
    \end{algorithmic}
\end{algorithm}

Note that the search path of $x$ after phase $1$ is a subset of the initial search path, in particular it contains the previously branching nodes $x_1, \dots, x_k$.
We illustrate the operation of SplayTT with an example in Figure~\ref{fig6}.

\medskip

As both phases involve a single pass through the search path, the entire process can be implemented in time linear in the length of the search path. The ZIG, ZIG-ZIG, and ZIG-ZAG steps take a constant number of rotations and oracle calls, the splaying in both phases thus clearly takes linear overall time. To find the branching nodes of the search path $P$ (Line~\ref{spll3}), we calculate the boundary sizes of the search path nodes as in \Cref{sec:rot-impl}, and identify the branching nodes using the characterization in \Cref{p:branch-equiv}. Note that, in the spirit of self-adjusting structures,  the algorithm does not need to know the global structure of $S$, and does not require persistent book-keeping between searches.

\begin{figure*}
  \centering
  \includegraphics[width=14cm]{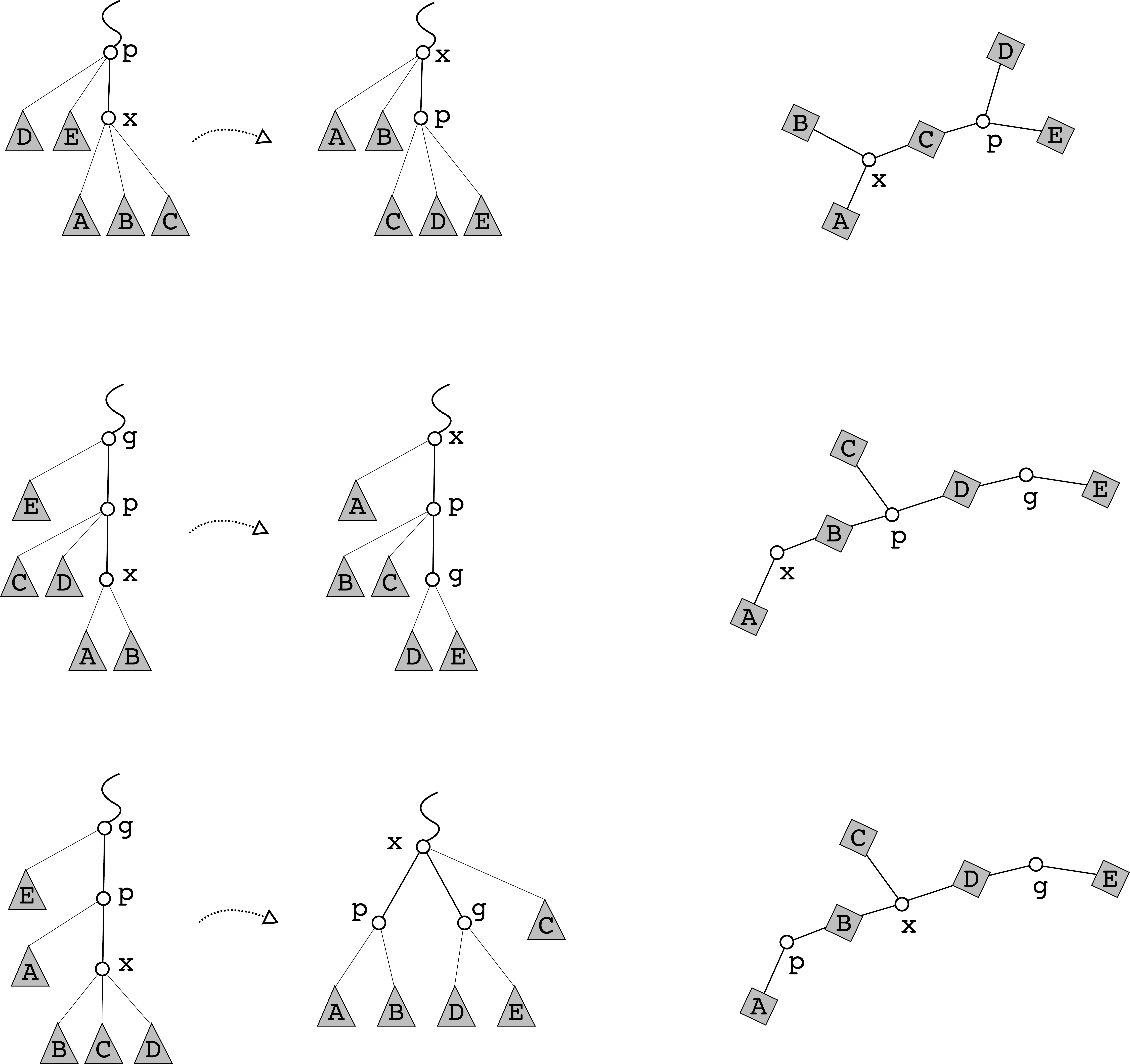}
  \caption{~~(\emph{left}) Splay transformation of search tree $T$ (from top to bottom) in ZIG, ZIG-ZIG, and ZIG-ZAG case.  ~~(\emph{right}) Underlying search space $S$. Triangles, resp.\ diamonds represent subtrees of $T$, resp.\ $S$. In contrast to classical Splay, the single-rotation (ZIG) case may be used also when $p$ is not the root. \label{fig4}}
\end{figure*}

\newpage

\section{Static optimality}\label{sec6}

As discussed in \S\,\ref{sec:intro}, Splay trees have several powerful adaptive properties, including static optimality. In this section we prove a similar property of our generalized SplayTT algorithm from \S\,\ref{sec5}. %

\restatethmc*

Let $R'$ be an \emph{optimal tree} for the search sequence $X$, i.e.\ so that $\cost_{R'}(X) = \OPT(X)$. By Lemma~\ref{lemdepth}, there is a Steiner-closed search tree $R$ on $S$ with $\cost_{R}(X) \leq 2 \cdot \OPT(X)$. It is important to note that neither $X$, nor $R',R$ are known to the algorithm, and are only used in the analysis. The initial state of SplayTT is an \emph{arbitrary} Steiner-closed STT on $S$. Let $T$ denote the state of the search tree before the operation search$(x)$. As before, let $T_x$ denote the subtree of $T$ rooted at $x$, and let us denote its node set $N_x = V(T_x)$. 

We define the \emph{node-potential} $\phi(x) =  - \min_{y \in N_x}{\bigl\{\hgt_R{(y)}\bigr\}}$ for all $x \in V(S)$. In words, the node potential is, without the minus sign, the smallest depth (in $R$) of $x$ or a descendant of $x$ (in T). The \emph{total potential} is  $\phi(T) = d \cdot \sum_{x \in V(T)} {\phi(x)}$, for a constant $d>0$ to be chosen later.\footnote{The potential function is inspired by a similar one suggested by T.~Saranurak for the analysis of classical Splay~\cite[\S\,3.2]{landscape}. For Splay (i.e.\ in the BST setting) the corresponding potential function can be seen to be essentially equivalent to the classical ``sum-of-logs'' potential of Sleator and Tarjan, but this seems not to be the case in the STT setting.}

We next bound the amortized cost of searching in $T$ by \Cref{alg:splay2} in terms of the cost in the optimal tree $R$. Let $T'$ denote the tree after searching $x$.

\begin{lemma} \label{lem:am}
With the above definitions, for some constant $c > 0$, $$\hgt_T(x) + \phi(T') - \phi(T) \leq c \cdot \hgt_R(x).$$
\end{lemma}

Before proving Lemma~\ref{lem:am}, we make some easy observations. 

\begin{lemma} \label{am:prop}
For an arbitrary tree $T$ of size $n$ and $\phi$ defined as above:
\begin{compactenum}[(i)]
\item The node potential satisfies $-\depth_R(x) \leq \phi(x) \leq -1$ for all $x \in V(T)$.
\item The total potential satisfies $-dn^2  \leq \phi(T) \leq 0$.
\item If $x$ is an ancestor of $y$ in $T$, then $\phi(x) \geq \phi(y)$.
\item For every subtree $S[A]$ of $S$, there is a unique $x \in A$ that minimizes $\hgt_R{(x)}$.\footnote{This property can be interpreted as saying that \emph{depth in $R$} gives a \emph{unique minimum coloring} (see e.g.~\cite{EvenSmorodinsky}) of the hypergraph formed by nodes and subtrees of $S$.}
\end{compactenum}
\end{lemma}

Lemma~\ref{lem:am} and Lemma~\ref{am:prop}(ii) together imply Theorem~\ref{thm4} by a standard amortization argument: since the work performed by SplayTT is at most $c' \cdot \hgt_T(x)$, for some constant $c' > 0$, by scaling $\phi$ with the same constant, we obtain the upper bound $cc' \cdot \depth_R(x) \in \fO(\depth_R(x))$ on the amortized cost of searching. This telescopes over the sequence of searches, yielding the bound $O(\OPT(X))$ on $\cost_T(X)$, with an additive term $O(n^2)$ due to the final potential.

In the remainder we prove the two lemmas.

\begin{proof}[Proof of \Cref{am:prop}] \ \\
\vspace{-0.15in}
\begin{compactenum}[(i)]
\item Follows from $1 \leq \hgt_R(x) \leq n$ and $\phi(x) \geq - \hgt_R(x)$ for all $x \in V(T)$.
\item Immediate by summing (i).
\item Follows as the minimum in $\phi(x)$ is taken over a larger set.
\item Suppose there are two nodes $x_1, x_2$ of minimum depth in $R$. Then they have a proper least common ancestor $x$ in $R$ and since $x$ separates $x_1,x_2$ in $S$, we have $x \in A$.\qedd
\end{compactenum}
\let\qed\relax
\end{proof}

\begin{proof}[Proof of \Cref{lem:am}]
\ \\
Both phases consist of ZIG, ZIG-ZIG, and ZIG-ZAG steps, as defined in Algorithm~\ref{alg:splay}.

We bound the \emph{increases} in node potential due to these elementary steps, denoted $\Delta \phi$, then sum up the individual increases, multiplied by $d$ to get a bound on the total increase in potential. We refer to node names as in Algorithm~\ref{alg:splay} and Figure~\ref{fig4}. By $\phi$ and $\phi'$ we denote respectively, the potential before and after the elementary step. By $N_x$ and $N'_x$ we denote the set of nodes in the subtree rooted at $x$ before, resp.\ after the elementary step.

\medskip

\noindent\textbf{ZIG:} Only nodes $p$ and $x$ change potential.
\begin{align*}
\Delta \phi & = \phi'(p) + \phi'(x) - \phi(p) - \phi(x)\\ & \leq 2 (\phi'(x) - \phi(x)) \tag*{$\triangleright$ by Lemma~\ref{am:prop}(iii)} \\ & \leq 3 (\phi'(x) - \phi(x)).\tag*{$\triangleright$ since $\phi'(x) \geq \phi(x)$}
\end{align*}
\textbf{ZIG-ZIG:} Only nodes $g$, $p$, and $x$ change potential.

\noindent \Cref{am:prop}(iv) implies that there is a unique node $y \in N'_x$ with minimum depth in $R$. Since $N_x \cap N'_g = \emptyset$ and $N_x \cup N'_g \subseteq N'_x$, either $y \notin N_x$, which implies $\phi'(x) \geq  \phi(x) + 1$, or $x \notin N'_g$, which implies $\phi'(x) \geq \phi'(g)+1$. Thus, $\phi(x) +\phi'(g) + 1 \leq 2\phi'(x)$. It follows that:
\begin{align*}
\Delta \phi & = \phi'(g) + \phi'(p) + \phi'(x) - \phi(g) - \phi(p) - \phi(x) \\
& =  \phi'(g) + \phi'(p) - \phi(p) - \phi(x) \tag*{$\triangleright$ since $\phi'(x) = \phi(g)$}\\
& \leq (2\phi'(x) - \phi(x) - 1) + \phi'(p) - \phi(p) - \phi(x) \tag*{$\triangleright$ by Lemma~\ref{am:prop}(iv)}\\
& \leq 3 (\phi'(x) - \phi(x)) - 1.\tag*{$\triangleright$ by Lemma~\ref{am:prop}(iii)}
\end{align*}

\noindent \textbf{ZIG-ZAG:} Only nodes $g$, $p$, and $x$ change potential.

\noindent Since $N'_g \cap N'_p = \emptyset$ and $N'_g \cup N'_p \subseteq N'_x$, by Lemma~\ref{am:prop}(iv), we have either $\phi'(x) \geq \phi'(g) + 1$, or $\phi'(x) \geq \phi'(p) + 1$, so $\phi'(g) + \phi'(p) + 1 \leq 2 \phi'(x)$. 
It follows that:
\begin{align*}
\Delta \phi & = \phi'(g) + \phi'(p) + \phi'(x) - \phi(g) - \phi(p) - \phi(x) \\
& =  \phi'(g) + \phi'(p) - \phi(p) - \phi(x) \tag*{$\triangleright$ since $\phi'(x) = \phi(g)$}\\
& \leq  \phi'(g) + \phi'(p) - \phi(p) - \phi(x) + (2\phi'(x)-\phi'(g) - \phi'(p) - 1) \tag*{$\triangleright$ by Lemma~\ref{am:prop}(iv)}\\
& \leq 2 (\phi'(x) - \phi(x)) - 1 \tag*{$\triangleright$ by Lemma~\ref{am:prop}(iii)}\\
& \leq 3 (\phi'(x) - \phi(x)) - 1.\tag*{$\triangleright$ since $\phi'(x) \geq \phi(x)$}
\end{align*}

\noindent Observe that the potential-increases due to the elementary steps telescope within the splaying operation, yielding an increase of $3(\phi'(x) - \phi(x)) - z$, where $\phi'(\cdot)$ and $\phi(\cdot)$ denote the node potentials before and after the splaying and $z$ is the number of ZIG-ZIG and ZIG-ZAG operations.

Recall that in phase $1$ we splay the branching nodes $x_k, \dots, x_1$ on the search path in bottom-to-top order. Each branching node $x_i$ is splayed up until it becomes a child of $x_{i-1}$, or, in case of $x_1$, until it becomes the root. After phase $1$, the search path of $x$ consists of $(x_1, \dots, x_k)$ possibly followed by a number of nodes (that were initially between $x_k$ and $x$). Let us now denote by $N_y$, resp.\ $N'_y$, the set of nodes in the subtree rooted at node $y$ before, resp.\ after phase $1$. We have $N'_{x_i} \subseteq N_{x_{i-1}}$, for all $i=k, \dots, 2$, so $\phi(x_{i-1}) \geq \phi'(x_i)$. By Lemma~\ref{am:prop}(i)(iii), $\phi(x_k) \geq \phi(x) \geq -\hgt_R(x)$, and $\phi'(\cdot) < 0$.

It follows that the total potential increase during phase $1$ is at most 
\begin{align*}
\Delta \phi^{(1)} & \leq d\Bigl(-z_1 + 3\sum_{i=1}^k{\bigl(\phi'(x_i) - \phi(x_i)\bigr)}\Bigr) \\
& \leq 3d\Bigl(\phi'(x_1) - \phi(x_k)\Bigr) - dz_1 \\
& \leq 3d \cdot \hgt_R(x) - dz_1, \tag*{$\triangleright$ by Lemma~\ref{am:prop}(i)(iii)}
\end{align*}
 where $z_1$ is the number of ZIG-ZIG and ZIG-ZAG steps in this phase.

\medskip

In phase 2 we splay $x$ to become the root, so the total potential increase in this phase is at most 
\begin{align*} 
\Delta \phi^{(2)} & \leq 3d \Bigl(\phi'(x) - \phi(x) \Bigr) -  dz_2 \\
& \leq 3d \cdot \hgt_R(x) - dz_2, \tag*{$\triangleright$ by Lemma~\ref{am:prop}(i)}
\end{align*}
 where $z_2$ is the number of ZIG-ZIG and ZIG-ZAG steps in this phase.

Denoting $z=z_1+z_2$, the total increase in potential during the two phases is therefore at most
$6d \cdot \hgt_R(x) - d z$.
Intuitively, we save constant potential on every ZIG-ZIG and ZIG-ZAG step, which we need to offset the actual cost of $\fO(\depth_T(x))$. In the first phase we may have the unfavorable situation that there are many branching nodes, with short \emph{odd-length} gaps between them, in which case we lose potential due to the many ZIG-steps. But then, since the branching nodes stay on the search path for the second phase, they form a long contiguous splaying path on which we make up for the loss.

For $i=2, \dots, k$, let $h_i$ denote the number of non-branching nodes between $x_{i-1}$ and $x_{i}$ at the beginning of phase $1$, let $h_1$ denote the non-branching nodes above $x_1$ on the path to the root, and let $h_{k+1}$ denote the number of nodes between $x$ and $x_k$. In the first phase, The number of ZIG-ZIG and ZIG-ZAG operations that $x_i$ participates in is $\lfloor h_i/2 \rfloor$, for all $i$.

Observe that $\hgt_T(x) = k + 1 + \sum_{i=1}^{k+1}{h_i}$. We thus have: 
\begin{align*}
z_1 ~~=~~ & \sum_{i=1}^k{\left\lfloor \frac{h_i}{2} \right\rfloor} ~\geq~ \sum_{i=1}^k{ \frac{h_i -1}{2} } ~=~ \frac{\hgt_T(x)  - h_{k+1}-1}{2} - k.
\end{align*}

In the second phase, $x_1, \dots, x_k$ are still ancestors of $x$ but no longer branching. In addition, the original nodes on the search path between $x_k$ and $x$ are still the ancestors of $x$. Thus, at the beginning of the phase $x$ has at least $k+h_{k+1}$ ancestors, and therefore:

\begin{align*}
z_2 ~~\geq~~ & \left\lfloor \frac{k + h_{k+1}}{2} \right\rfloor ~\geq~ \frac{k+h_{k+1}-1}{2}.
\end{align*}

In both phases together, observing that $k \leq \hgt_T(x)/2$, we have:
\begin{align*}
z ~~=~~ & z_1+z_2 ~\geq~ \frac{\hgt_T(x)-k}{2}-1 ~\geq~ \frac{\hgt_T(x)}{4} - 1.
\end{align*}
The total increase in potential is $6d \cdot \depth_R(x) - d/4 \cdot \depth_T(x) - O(1)$.

Choosing $d=4$, we obtain $\hgt_T(x) + \phi(T') - \phi(T) \leq 24 \cdot \hgt_R(x) + O(1) \leq \fO(\hgt_R(x))$. \qedd

\end{proof}

\paragraph{Remark.} The additive term $O(n^2)$ is an artifact of the proof technique, arising as an upper bound on the total potential $|\phi(T)|$. A finer bound of $|\phi(T)| \leq d \cdot \sum_{x \in V(S)}{\hgt_R(x)}$ is immediate from Lemma~\ref{am:prop}(i)(ii). This shows that (1) the $O(n^2)$ upper bound is rather loose, unless the optimal tree $R$ is very unbalanced, and (2) if every node is searched \emph{at least once} in $X$, then $|\phi(T)| \leq d \cdot \OPT(X)$, and the additive term can be absorbed into $\fO(\OPT)$. %

\vfill

\begin{figure*}[h]
  \centering
  \includegraphics[width=8cm]{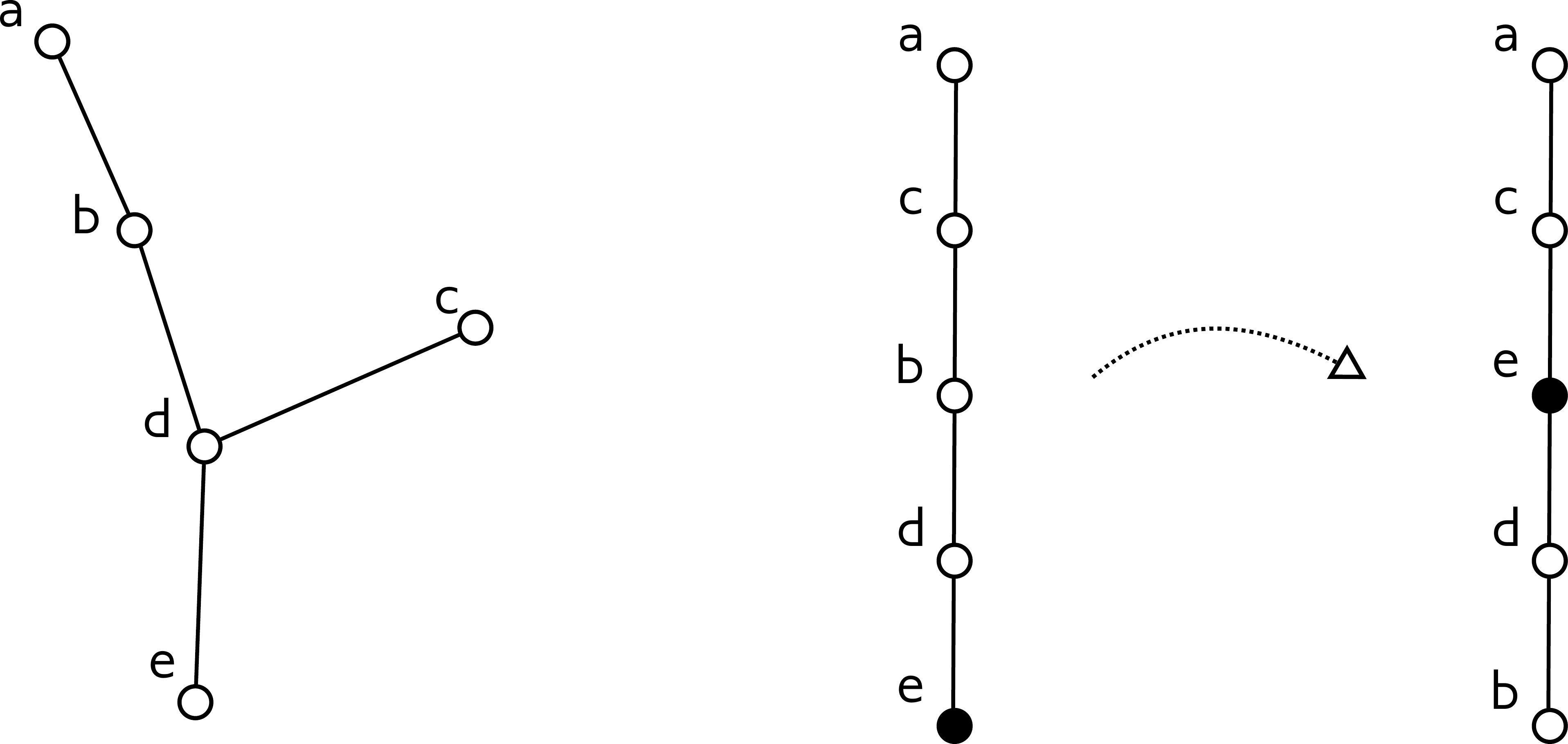}
  \caption{~~Example showing that na\"{i}vely splaying in a Steiner-closed tree may destroy the property. (\emph{left}) Underlying tree $S$. (\emph{right}) Splaying node $e$. After the first (ZIG-ZIG) step, the search tree is no longer Steiner-closed. \label{fig7}}
\end{figure*}

\begin{figure*}
  \centering
  \includegraphics[width=16cm]{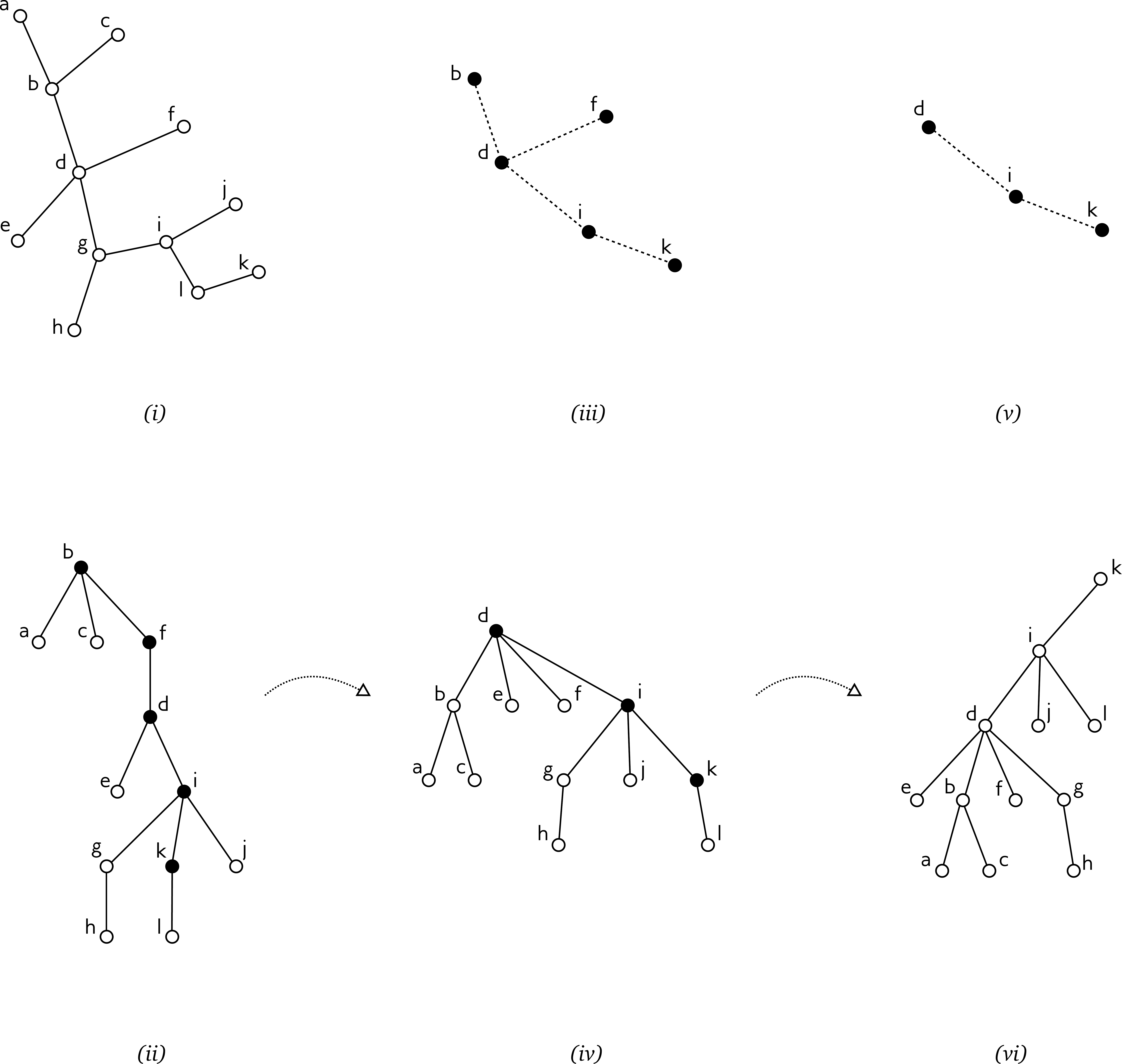}
  \caption{~~Example operation in SplayTT: searching for node $k$. \\
  (\emph{i})~Underlying tree $S$.\\ (\emph{ii})~Steiner-closed search tree $T$ on $S$ with search path of node $k$ shaded. \\(\emph{iii}) Nodes on search path of $k$ shown within $S$, with paths between nodes indicated with dotted line. Observe that $d$ is the only branching node on the search path.\\ (\emph{iv})~Result of the first phase of splaying. The branching node $d$ is splayed up to the root with a single ZIG-ZAG step. Remaining search path of $k$ shaded. \\(\emph{v})~Remaining search path of $k$ shown within $S$. Observe that there are no remaining branching nodes.\\ (\emph{vi})~Result of the second phase of splaying. Node $k$ is splayed up to the root with a single ZIG-ZIG step. \label{fig6}}
\end{figure*}

\newpage

\newpage

	\newpage 
	
	\bibliography{paper}{}
	\bibliographystyle{alpha}
	
	\newpage

\begin{appendix}
\section{Leaf centroid}\label{appa}
\begin{lemma}
	Every tree $S$ with $n \ge 3$ nodes has a leaf centroid, that can be found in $\fO(n)$ time.
\end{lemma}
\begin{proof}
	Let $L$ be the set of leaves of $S$ and let $\ell = |L|$. For an arbitrary node $v \in V(S)$, let $t(v) = (\ell(v), k(v))$ be the lexicographically largest tuple $(|L \cap V(C)|, |V(C)|)$ among the connected components $C$ of $S \setminus v$. We claim that if $\ell(v) > \floor{\frac{\ell}{2}}$, then there is a neighbor $u$ of $v$ such that $t(u)$ is lexicographically smaller than $t(v)$. By repeatedly moving to such a neighbor we will find a node $v$ with $\ell(v) \le \floor{\frac{\ell}{2}}$, as there are only finitely many possible values for $t(v)$. 
	Each connected component $C$ of $S \setminus v$ %
	has at most one leaf that is not a leaf of $S$, namely the node that is adjacent to $v$. %
	It follows that $v$ is a leaf centroid. %
	
	We now prove the claim. Let $v \in V(S)$ with $\ell(v) > \floor{\frac{\ell}{2}}$. Let $C$ be the connected component of $S \setminus v$ with $|L \cap V(C)| = \ell(v) > \floor{\frac{\ell}{2}}$ and let $u$ be the neighbor of $v$ in $S$ which is contained in $V(C)$. Let $S' = S[V(S) \setminus V(C)]$.
	
	Clearly, $S'$ is one connected component of $S \setminus u$. First, consider the case when $u$ has degree at least 3 in $S$. Then $S \setminus u$ has at least two connected components besides $S'$, each containing at least one leaf of $S$. Each of these components thus contains at most $|L| - |V(S')\cap L| - 1 \le |V(C)| - 1 < \ell(v)$ leaves of $S$. Moreover, $|S' \cap L| \le |L| - |L \cap V(C)| < \ell - \floor{\frac{\ell}{2}} = \ceil{\frac{\ell}{2}} \le \ell(v)$. Thus, we have $\ell(u) < \ell(v)$.
	
	Second, consider the case when $u$ has degree 2 in $S$. Then $S \setminus u$ has two connected components, $S'$ and, say, $S''$. We have $|L \cap V(S'')| = |L \cap V(C)| > \floor{\frac{\ell}{2}}$, i.e.\ $S''$ has more than half the leaves of $S$, and in particular $|L \cap V(S'')| > |L \cap V(S')|$. Thus, by definition, $t(u) = (|L \cap V(S'')|, |V(S'')|)$. But $|V(S'')| = |V(C)| - 1 = k(v) - 1$. As such, $t(u)$ is lexicographically smaller than $t(v)$.
	
	The tuples $t(v)$ can be computed for all nodes $v$ in $O(n)$ time via a depth-first traversal of $S$. When deciding where to move from $v$, we need to query the tuples of neighbors of $v$, and pick the smallest. For each neighbor $u$ we charge the query to the (oriented) edge $(v,u)$ of $S$. As we cannot revisit a vertex, each edge is charged at most once. \qedd
\end{proof}

\newpage
\section{Improvement of Lemma~\ref{p:depth_approx}}\label{appb}

We describe a modification of Algorithm~\ref{alg:trans-st}, slightly improving the depth-approximation factor. 
The idea is to execute the root replacement of \Cref{alg-line10} only when necessary. Suppose we call $\Call{Fix}{T, x}$. Let $A = V(T_x)$ and suppose $|\delta(A)| = k$, and $x$ is a $k$-admissible root of $A$ (see \S\,\ref{sec:k-cut-opt-find}). Then, we can simply recurse on the children, as in Lines~\ref{alg-line:just-recurse-1}--\ref{alg-line:just-recurse-2}. From the discussion in \S\,\ref{sec:k-cut-opt-find} it is clear that the result is still a $k$-cut tree.

Now suppose that $x$ is not a $k$-admissible root, i.e.\ $x \notin \ch(\delta(A))$. Then, in \Cref{alg-line:leaf-centroid}, we choose $v$ to be the leaf centroid of the tree $\ch(\{x\} \cup \delta(A))$ instead of $\ch(\delta(A))$. Observe that the maximum boundary size of components of $S[A] \setminus v$ is $\floor{\frac{k+1}{2}}+1 = \ceil{\frac{k}{2}}+1 \le k$, so the algorithm is still correct. 

We consider how the depth of each node changes by this modified transformation. Let $C$ be the component of $S[A] \setminus v$ that contains $x$. When changing $T_x$ to $T'_v$, a node $u$ in $C$ may either gain $v$ as an ancestor, or otherwise only some nodes on the search path of $u$ are permuted. The depth of $u$ thus increases by at most one. All other nodes in $A \setminus V(C)$ lose $x$ as an ancestor and may, at worst, gain $v$ as an ancestor, their depth thus does not increase.

We now argue that, while following the search path from the root to some node $u$ in the final tree $T^*$ (which corresponds to the recursion tree of $\Call{Fix}{}$), the depth of a node can only increase every $\floor{k/2}+1$ steps. We already observed that an increase in depth can only happen in the following situation: $x$ is not a $k$-admissible root of $A = V(T_x)$ and $u$ is in the same connected component $C$ of $S[A] \setminus v$ as $x$. We know that $v \in \delta(C)$. As $x \in V(C)$, and because of our choice of $v$, there are at most $\floor{\frac{k+1}{2}}-1 = \ceil{k/2}-1$ other nodes in $\delta(C)$. This means that after executing the non-recursive portions of $\Call{Fix}{T,x}$ (which replaced $x$ with $v$) and then $\Call{Fix}{T',x}$ (which cannot change the tree), the depth cannot increase for the next $\floor{k/2}$ calls on the path to $u$.

Thus, the depth of each node increases by a factor of at most $1 + 1/\floor{k/2}$. We obtain the following. %

\begin{lemma} [Strengthening of Lemma~\ref{p:depth_approx}]\label{p:depth_approx-improved}
	Given an STT $T$ on $S$ and $k \ge 2$, we can find a $k$-cut STT $T^*$ on $S$, so that $\depth_{T^*}(x) \le (1+\epsilon_k) \cdot \depth_T(x)$ for all $x \in V(S)$, where $$\varepsilon_k = \frac{1}{\left\lfloor \frac{k}{2} \right\rfloor}.$$
\end{lemma}

\end{appendix}

\end{document}